\documentclass[11pt]{article}
\usepackage{amsmath}
\usepackage{amsfonts}
\usepackage{amsthm}
\usepackage{amssymb}
\usepackage[mathscr]{euscript}
\usepackage[margin=1in]{geometry}
\usepackage{hyperref}
\usepackage{graphicx}
\usepackage{caption}

\def\>{\rangle}
\def\<{\langle}
\def\ug{\mathbb{U}}
\def\cE{\mathcal{E}}
\def\cP{\mathcal{P}}
\def\cQ{\mathcal{Q}}
\def\cC{\mathcal{C}}

\def\SL{\mathrm{SL}}
\def\LT{\triangle}
\def\UT{\bigtriangledown}

\def\be{\begin{equation}}
\def\ee{\end{equation}}
\def\bea{\begin{eqnarray}}
\def\eea{\end{eqnarray}}
\def\etwirl{\varepsilon_{\mbox{\scriptsize bilateral}}}
\def\echannel{\varepsilon_{\mbox{\scriptsize channel}}}

\def\>{\rangle}
\def\<{\langle}
\def\ug{\mathbb{U}}
\def\cE{\mathcal{E}}
\def\cP{\mathcal{P}}
\def\cQ{\mathcal{Q}}
\def\cC{\mathcal{C}}
\def\be{\begin{equation}}
\def\ee{\end{equation}}
\def\bea{\begin{eqnarray}}
\def\eea{\end{eqnarray}}
\newcommand{\comment}[1]{}

\newcommand{\set}[1]{{\left\{#1\right\}}}    

\newcommand{\GF}{\mbox{\rm GF}}

\def\ket#1{ | #1 \rangle}

\def\tr{ {\rm{Tr }}}

\renewcommand{\tfrac}[2]{\textstyle \frac{#1}{#2}}

\newtheorem{theorem}{Theorem}

\newtheorem{lemma}{Lemma}
\newtheorem{defn}{Definition}
\newtheorem{cor}[theorem]{Corollary}
\newtheorem{prop}[theorem]{Proposition}

\usepackage{color,graphicx}

\begin{document}
\title{{\bf Near-linear constructions of exact unitary 2-designs}}

\author{Richard Cleve%
\thanks{Institute for Quantum Computing and School of Computer Science, University of Waterloo}
\thanks{Canadian Institute for Advanced Research}
\and Debbie Leung%
\thanks{Institute for Quantum Computing and Department of Combinatorics and Optimization, University of Waterloo}
$^{\dag}$
 \and Li Liu$^{*}$ \and Chunhao Wang$^{*}$}

\date{\empty}

\maketitle

\begin{abstract}
A unitary 2-design can be viewed as a quantum analogue of a 2-universal hash function: it is indistinguishable from a truly random unitary by any procedure that queries it twice.
We show that exact unitary 2-designs on $n$ qubits can be implemented by quantum circuits consisting of $\widetilde{O}(n)$ elementary gates in logarithmic depth.
This is essentially a quadratic improvement in size (and in width times depth) over all previous implementations that are exact or approximate (for sufficiently strong approximations).
\end{abstract}


\section{Introduction}

The uniform distribution on the group consisting of all unitary operations acting on $n$ qubits is captured by the \textit{Haar measure}, which is the unique measure that is
invariant under left and right multiplication by any group element.
Haar-random unitaries, by their symmetries, facilitate many analyses
in quantum information \cite{Hastings09,HHYW08,HLSW04,HLW06,HW08}.
However, Haar-random unitaries have very high computational complexity,
in that most of them cannot be efficiently implemented or reasonably
well approximated by circuits of size polynomial in the number of
qubits.  They require many bits to describe.  They also require a lot
of randomness to sample.

Unitary 2-designs are probability distributions on finite subsets of
the unitary group that have some specific properties in common with
the Haar measure.  Several common definitions for unitary 2-designs
have been proposed and studied, each revolving around a specific
property or application, and appropriate notion of approximation
\cite{DLT02,DCEL09,HL09}.  These 2-designs are computable by
polynomial-size circuits with short specifications and low sampling
complexity.

We focus on \textit{exact} unitary 2-designs.  In the exact case, we
will see that several commonly used definitions can be shown to be
equivalent to each other.  One particularly natural definition is that
they are \textit{two-query indistinguishable} from Haar-random
unitaries.  Imagine a game where, at the flip of a coin, $U$ is
sampled either according to the Haar measure or with respect to the
unitary 2-design.  A two-query distinguishing procedure can make two
queries to $U$ (each being either in the forward direction as $U$ or
in the reverse direction as $U^{\dag}$) as well as other quantum
operations that do not depend on $U$ and then outputs a bit.  A
unitary 2-design has the property that no two-query distinguishing
procedure can distinguish between the Haar-random case and the
2-design case with probability greater than $1/2$.  By this
definition, a unitary 2-design is a quantum analogue of a 2-universal
hash function~\cite{CarterWegman1977} (or, more precisely, 2-universal hash permutation).
We will show in Section~\ref{sec:definitions} that this definition is equivalent to previous definitions, including those based on \textit{bilateral twirling} \cite{DLT02} and \textit{channel twirling} \cite{DCEL09}.

\subsection{Previous work}

The uniform distribution on the Clifford group has been shown to be an
exact unitary 2-design in the sense of bilateral twirling \cite{DLT02} and channel twirling \cite{DCEL09}.  This
implies that the circuit complexity is $O(n^2/\log n)$ where the gates
are one- and two-qubit gates from the Clifford group~\cite{AG2004}.
Moreover, the sampling cost is $O(n^2)$ random bits of entropy.  In
the context of bilateral twirling, \cite{HL09} shows that a
certain process of random circuit generation 
(introduced in~\cite{EWSLC03}) yields
$\etwirl$-approximate unitary 2-designs of size $O(n(n+\log
1/\etwirl))$, where $\etwirl$ measures the distance of the resulting
operation from the ideal one.

Another construction~\cite{DCEL09}
yields circuits of size $O(n \log 1/\echannel)$ for a notion of
approximation that is natural for channel twirling; however, it has 
been pointed out that this notion of approximation could 
(at least conceivably) incur a blow-up by a factor that is exponential 
in $n$ in the bilateral twirl
context (see, e.g., Section~2 of~\cite{HL09} and Section~1.1
of~\cite{BF2013} for discussion about this).
For the more general setting, as far as we know, we might 
need $\echannel \le \etwirl/2^n$ --- so the circuit size becomes $O(n(n+ \log
1/\etwirl))$.

For exact unitary 2-designs as well as approximations
of them related to bilateral twirling, all of the above constructions
incur circuits of size $\widetilde{\Omega}(n^2)$ and require
$\Omega(n^2)$ random bits of entropy.

Reference \cite{Chau2005} proves that there exists a small subgroup of
the Clifford group that gives rise to an exact unitary 2-design that
uses approximately $5n$ random bits of entropy.  However, the circuit
complexity for this construction is unknown, beyond the $O(n^2/\log
n)$ bound that holds for any Clifford operation.  References
\cite{GAE2007} and \cite{RS2009} study the necessary and sufficient
entropy for exact and approximate unitary $2$-designs.  Approximately
$4n$ random bits of entropy are necessary.

\subsection{New results}

We give three constructions of \textit{exact} unitary 2-designs on $n$
qubits that have the following quantum gate costs (number of one- and
two-qubit gates):
\begin{itemize}
\item
$O(n \log n \log\log n)$ gates (all Clifford gates) for infinitely many $n$, assuming the extended Riemann Hypothesis is true.
\item
$O(n \log n \log\log n)$ gates (including non-Clifford gates) for all $n$, unconditionally.
\item
$O(n \log^2 n \log\log n)$ gates (all Clifford gates) for all $n$, unconditionally.
\end{itemize}
The circuits for the first two constructions can be organized so as to
perform their computation in $O(\log n)$ depth; the third in
$O(\log^2 n)$-depth (using the fact that efficient
multiplication/convolution algorithms require only $O(\log
n)$-depth~\cite{schonhage1977}).  These results are near optimal -- 
in Appendix \ref{appendix:optimality}, we show that for any 
unitary $2$-design (exact or approximate under Definition~\ref{def:2queryu} or 
\ref{def:btwirl}), a high probability set
of the unitaries have size $\Omega(n)$ and depth $\Omega(\log n)$.


All three constructions above use $5n$ bits of randomness (more
precisely, they sample from a uniform distribution on a set of size
$2^{5n} -2^{3n}$).  They all consist of unitaries in the Clifford
group (even in the second construction, non-Clifford gates are
used to compute Clifford unitaries efficiently).  The circuits use
$\widetilde O(n)$ ancilla qubits (where each ancilla qubit is
initially in state $\ket{0}$ and is restored to this state at the end
of the computation).  Finally, the cost of the classical process that
samples these unitary 2-designs (outputs a description of the quantum
circuit) is polynomial in $n$.  The cost is dominated by the
complexity of computing square roots in the finite field $\GF(2^n)$.

It should also be noted that our Definition~\ref{def:2queryu} is a new
characterization of unitary 2-designs in terms of 2-query
indistinguishability that may be of independent interest.

\subsection{Significance of the new constructions}

Since our constructions yield exact unitary 2-designs, they are
automatically valid for all notions and definitions of approximate
unitary 2-designs.
Our construction thus achieves the minimum known circuit size, depth,
and sampling complexity simultaneously, among both exact and
approximate unitary 2-designs.

Exact $2$-designs offer other advantages.  Besides the original
operational applications of bilateral and channel twirling,
$2$-designs have appeared in second moment analysis.  For example,
they arise in \cite{HHYW08}, where results are obtained about the
decoupling of two quantum systems and quantum channel capacities.  An
exact 2-design can be used in a ``plug-and-play" manner.  For example,
there exists an encoding operation in any unitary $2$-design that,
when concatenated with an appropriate inner code, achieves the quantum
channel capacity.  Thus, our results {\em automatically} imply the
existance of such encoding circuits of $O(n \log^2 n \log\log n)$
Clifford gates and depth $O(\log^2 n)$.

In some applications such as decoupling, the distance from an exact
$2$-design is amplified by a dimensional factor that can be
exponential in $n$ (for example, Theorem~1 in~\cite{SDTR2013}). Using
our exact construction, such error term vanishes exactly, so does the
issue of the exponential amplification of errors.  Thus our results
yield potentially tighter bounds while maintaining a circuit size of
$\widetilde{O}(n)$.

Prior to our work, \cite{BF2013} constructs a method to generate
random circuits of size $O(n \log^2n)$ and depth $O(\log^3n)$ that
does not give rise to a $2$-design, yet achieves decoupling and
provides small encoding circuits for quantum error correcting codes
\cite{BF2013b}.  The advantage in their approach is that no ancillas
are needed, and the circuit may model some random physical processes.
However, the depth is higher, and a substantial amount of analysis is
required in the aforementioned references to show that the
construction and circuit size indeed achieve the tasks with the
desired accuracy.  Adapting their construction to other applications
may also require additional analysis.


\section{Definition of a unitary 2-design}\label{sec:definitions}

We first discuss several definitions that are equivalent to the concept of
unitary $2$-designs.

Let $\ug_N$ denote the group of $N \times N$ unitary matrices.  We are
interested in distributions over $\ug_N$.  The Haar measure on $\ug_N$
is the unique measure on $\ug_N$ that is invariant under left and
right multiplication by any $U \in \ug_N$.  We denote the Haar measure
by $\mu(U)$.  Let $\cE = \set{p_i,U_i}_{i=1}^k$ denote a finite
ensemble of unitary matrices $U_1, U_2, \cdots, U_k \in \ug_N$ where
$p_i \geq 0$ and $\sum_i p_i = 1$.

Sampling from the Haar measure is a powerful technique in quantum
information theory.  Sometimes, we use a physical procedure that
averages over such random choices of unitary transformation (for
example \cite{DLT02,DCEL09}). Some other times, we have a randomized
argument, for example, in the proof of quantum channel capacity
\cite{Devetak05,HHYW08,Lloyd97,Shor02}, in which the average 
performance over all possible unitary encodings is evaluated.

We are interested in contexts in which such sampling from the Haar
measure can be replaced by sampling from a finite ensemble 
$\cE = \set{p_i,U_i}_{i=1}^k$ of unitary matrices.  This can reduce the
required resources such as shared randomness, communication to
implement the random unitary, as well as the computational complexity
of implementing the randomly chosen unitary.  We now discuss several
of these circumstances.

The first context is concerned with the expected value of polynomials of
the entries of unitary matrices drawn according to some distribution.
This definition is essentially the original definition of 
unitary 2-design in~\cite{DCEL09}, and is useful for proving results in
other contexts.

\begin{defn}
\label{def:u2design}
We say that $\cE$ is \emph{degree-$2$ expectation preserving} if, for every polynomial $\gamma(U)$
  of degree at most $2$ in the matrix elements of $U$ and at most $2$
  in the complex conjugates of those matrix elements, 
\be 
\sum_{i=1}^k p_i \, \gamma(U_i) = \int d\mu(U) \, \gamma(U) \,.
\label{eq:u2design}
\ee
\end{defn}

\noindent In Eq.~(\ref{eq:u2design}) and throughout the paper, an integral 
written without a specific domain is taken over~$\ug_N$.

The second context is concerned with distinguishing whether a random
sample $U$ is drawn from the Haar measure or from the ensemble $\cE$,
when an arbitrary distinguishing circuit is allowed to make a total of
at most two queries of $U$ or $U^\dagger$.  The most general circuit $\mathcal C$
of this form is depicted in Figure~\ref{fig:def-2-design}.
\begin{figure}[h!]
\centering
\includegraphics[width=0.35\textwidth]{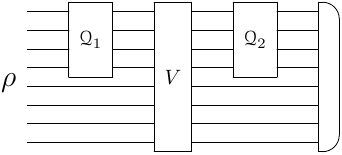}
\caption{\small Illustration of a \textit{$2$-query distinguishing
    circuit} $\mathcal C$.  The first query $\mathscr{Q}_1$ can be $U$
  or $U^\dagger$, likewise for the second query $\mathscr{Q}_2$.  The
  initial state $\rho$ is arbitrary, $V$ is an arbitrary unitary, and
  the final measurement outputs one bit but is otherwise
  arbitrary.}\label{fig:def-2-design}
\end{figure}


The circuit $\mathcal C$ starts with an arbitrary initial state $\rho$ (a
positive semidefinite matrix of trace $1$).  Then, the first query, an
arbitrary operation $V$, the second query, and an arbitrary final
measurement that outputs one bit are applied in order.  We call any
such circuit a $2$-query distinguishing circuit.  
If $U$ is drawn from either ensemble, denote the quantum state right before the
measurement as $\eta_2(\mathcal C,U)$.  If $U$ is drawn from $\cE$, the density 
matrix in $\mathcal C$ before the final measurement is 
$\sum_{i = 1}^k p_i\eta_2(\mathcal C, U)$; similarly, if $U$ is drawn from the  
Haar measure, the density matrix before the final measurement is 
$\int d\mu(U) \eta_2(\mathcal C, U)$. The output bit of the circuit $\mathcal C$
has the same distribution regardless of which ensemble $U$ is sampled from, if 
and only if the above two density matrices are equal. The following definition 
describes ensembles that cannot be distinguished by any 2-query distinguishing circuit $\mathcal C$.

\begin{defn}
\label{def:2queryu}
We say that $\cE$ is \emph{$2$-query indistinguishable}, if, for any distinguishing
circuit $\mathcal C$ making up to two queries of a random unitary or its 
adjoint, 
\be \sum_{i=1}^k \; p_i \, \eta_2(\mathcal C,U_i) = \int d\mu(U) \, \eta_2(
\mathcal C,U) \,.\ee
\end{defn}

The next context is a special case of the scenario depicted in 
Figure~\ref{fig:def-2-design}, where $U$ is queried twice in parallel, as 
illustrated in Figure~\ref{fig:bilateral}.
\begin{figure}[h!]
  \centering
  \includegraphics[width=0.2\textwidth]{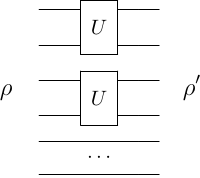}
  \caption{\small Illustration of the \textit{bilateral twirl}:
	querying $U \otimes U$. The initial state $\rho$ is 
  arbitrary.}\label{fig:bilateral}
\end{figure}
Consider bipartite operations in which two disjoint systems undergo
the same unitary transformation drawn according to some distribution.
These operations are sometimes called bilateral twirls \cite{BDSW96,DLT02}.
The \textit{$\cE$ bilateral twirl} is defined as the quantum operation 
\be
{\cal T}_\cE(\rho) = \sum_{i=1}^k \, p_i \, (U_i \otimes U_i) \, 
\rho \; (U_i^\dagger \! \otimes U_i^\dagger) \,.
\label{eq:btwirle}
\ee
The \textit{full bilateral twirl} is defined as the quantum operation 
\be
{\cal T}_\mu(\rho) = \int \, d \mu(U) \; ( U \otimes U ) \, \rho \; (
U^\dagger \! \otimes U^\dagger) \,.
\label{eq:btwirlf}
\ee
The full bilateral twirl is motivated
operationally \cite{BDSW96,DLT02} and it appears in various mathematical
proofs in quantum information \cite{HHYW08,SDTR2013}.  
Definition
\ref{def:btwirl} describes ensembles that derandomize the full bilateral
twirl.
\begin{defn}
\label{def:btwirl}
We say that the ensemble $\cE$ \emph{implements the full bilateral twirl} if
${\cal T}_\mu(\rho) = {\cal T}_\cE(\rho)$ for all $\rho$.
\label{eq:btwirl}
\end{defn}

The fourth context is concerned with the task of converting any
quantum channel into a depolarizing channel of the same average
fidelity.  This conversion has many important applications, for
example, {\em benchmarking} (for estimating average channel fidelity)
of quantum devices \cite{DCEL09} and error estimation (for detecting
eavesdropping) in quantum key distribution \cite{Chau2005}.

Let $\Lambda$ be any quantum channel that maps $N$-dimensional quantum
states to $N$-dimensional quantum states.  An $\cE$-channel-twirl of
$\Lambda$, denoted by $\mathbb{E}_\cE(\Lambda)$, is defined as the
quantum channel that acts as 
\be 
\mathbb{E}_\cE(\Lambda): \rho \mapsto \sum_{i=1}^k \, p_i \,
U_i^\dagger \Lambda(U_i \, \rho \, U_i^\dagger) \,U_i  \,.
\label{def:e-twirl-channel}
\ee
In other words, a random change of basis is applied to the system
before the channel $\Lambda$ acts and it is reverted afterwards.
A full channel twirl of a quantum channel $\Lambda$ is given by
\be 
\mathbb{E}_\mu(\Lambda): 
\rho \mapsto 
\int d\mu(U) \, U^\dagger \Lambda(U \, \rho \, U^\dagger) \, U \,.
\label{def:haar-twirl-channel}
\ee

\begin{defn}
  \label{def:channeltwirl}
  We say that $\cE$ \emph{implements the full channel twirl} if
  $\mathbb{E}_\cE(\Lambda) = \mathbb{E}_\mu(\Lambda)$ for all quantum
  channels $\Lambda$.
\end{defn}

\noindent Lemma~\ref{lem:equivdef} below states that these four relationships 
between ensembles and the Haar measure are equivalent.
Thus, we can think of an ensemble satisfying one of the conditions in 
alternative ways.
\vspace*{12pt}
\begin{lemma}
\label{lem:equivdef}
Let $\cE$ be any ensemble of unitaries in $\ug_N$.  
Then, the following are equivalent: 
\begin{itemize}
  \setlength\itemsep{0em}
  \item[(1)] $\cE$ is degree-$2$ expectation preserving. 
  \item[(2)] $\cE$ is $2$-query indistinguishable. 
  \item[(3)] $\cE$ implements the full bilateral twirl. 
  \item[(4)] $\cE$ implements the full channel twirl.
\end{itemize}
\end{lemma}
\vspace*{12pt}


\noindent The following corollary of Lemma \ref{lem:equivdef} is not obvious
from Definitions \ref{def:btwirl} and \ref{def:channeltwirl} alone.

\begin{cor}
\label{cor:dagger}
For $\cE = \set{p_i,U_i}_{i=1}^k$, let $\cE^\dagger {:}{=}
\set{p_i,U_i^\dagger}_{i=1}^k$. 
\begin{itemize}
  \setlength\itemsep{0em}
  \item[(a)] $\cE$ implements the full bilateral twirl if and only if $\cE^\dagger$ does. 
  \item[(b)] $\cE$ implements the full channel twirl if and only if $\cE^\dagger$ does.
\end{itemize}
\end{cor}


We note that additional definitions have been discussed in 
\cite{GAE2007,RS2009,HL09,Low10}.  
Several parts of Lemma \ref{lem:equivdef} have been proved in
literature \cite{DCEL09,HL09,Low10}.  In particular, \cite{Low10} 
relates definitions (1), (3), and (4) with bounds on the 
approximations.  
We provide a complete (alternative) proof of Lemma~\ref{lem:equivdef}
and Corollary~\ref{cor:dagger} 
in Appendix~\ref{appendix:definitions}.

Due to Lemma~\ref{lem:equivdef}, when we do not need to specify
the context, we just call an ensemble satisfying any one of the four 
conditions a ``unitary $2$-design.''  

\section{Pauli mixing implies a unitary 2-design}\label{sec:pauli-mixing}

We describe a simple sufficient condition for $\cE$ to be a
unitary $2$-design.

We begin by reviewing some basic definitions and terminology associated with the Pauli group.
Let 
$X = \bigl(\begin{smallmatrix}
	  0 & 1\\ 1 & 0
\end{smallmatrix}\bigr)$, 
$Y = \bigl(\begin{smallmatrix}
	  0 & -i\\ i & 0
\end{smallmatrix}\bigr)$, and
$Z = \bigl(\begin{smallmatrix}
	  1 & 0\\ 0 & -1
\end{smallmatrix}\bigr)$ 
denote the $2{\times}2$ Pauli matrices.
For any $a \in \{0,1\}^n$, define $X^a = X^{a_1} \otimes \cdots \otimes X^{a_n}$ and $Z^a = Z^{a_1} \otimes \cdots \otimes Z^{a_n}$.

\begin{defn}
The \textit{Pauli group} $\cP_n$ consists of all operators of the form $i^k X^a Z^b$, where $k \in \{0,1,2,3\}$ and $a, b \in \{0,1\}^n$.
Let $\cQ_n = \cP_n/\{\pm 1,\pm i\}$, the quotient group that results from disregarding global phases in $\cP_n$ (each element of $\cQ_n$ can be represented as $P_{a,b} = X^a Z^b$).
We call $P_{0,0} = I$ the trivial Pauli.
\end{defn}
Let 
$H = \frac{1}{\sqrt{2}}\bigl(\begin{smallmatrix}
	  1 & \ 1\\ 1 & -1
\end{smallmatrix}\bigr)$
(the $2{\times}2$ Hadamard matrix), $S = \bigl(\begin{smallmatrix}
	  1 & 0\\ 0 & i
\end{smallmatrix}\bigr)$
(the phase gate), and
\begin{align}
\mbox{CNOT} = \begin{pmatrix}
1 & \ 0 & \ 0 & \ 0 \\ 
0 & \ 1 & \ 0 & \ 0 \\ 
0 & \ 0 & \ 0 & \ 1 \\ 
0 & \ 0 & \ 1 & \ 0 
\end{pmatrix}.
\end{align}
\begin{defn} \label{def:clifford}
The \textit{Clifford group} $\cC_n$ is the set of all unitary matrices that permute the elements of $\cP_n$ (and thus $\cQ_n$) under conjugation.  
\end{defn}

The Clifford group $\cC_n$ contains the $H$, CNOT, and $S$ gates, and
they form a generating set \cite{Gottesman98}.  Conjugating the
elements in $\cP_n$ by some $U\in \cC_n$ gives a permutation on
$\cP_n$; this also induces a permutation $\pi_U$ on~$\cQ_n$.

\begin{defn}
Consider an ensemble $\cE = \set{p_i,U_i}_{i=1}^k$ of unitary
matrices $U_1, U_2, \cdots, U_k$ in the Clifford group $\cC_n$. 
We say that $\cE$ is {\em Pauli mixing}, if for all $P \in \cQ_n$ such that 
$P \neq I$, the distribution $\set{p_i, \pi_{U_i}(P)}$ is uniform 
over $\cQ_n \backslash \set{I}$.  
\end{defn}

For any ensemble $\cE = \set{p_i,U_i}_{i=1}^k$, let $\cE_\cQ =
\set{2^{{-}2n}p_i,U_i R_j}_{i=1,j=1}^{k,2^{2n}}$ where $R_j$ ranges over
all elements in $\cQ_n$.  Intuitively, $\cE_\cQ$ is the ensemble where
a random element drawn from $\cE$ is preceded by a random Pauli
operation drawn from $\cQ_n$.  

Pauli mixing by $\cE$ is a sufficient condition for the ensemble $\cE_\cQ$ of 
Clifford unitaries to be a unitary $2$-design. More specifically, we have the 
following lemma.
\begin{lemma}
\label{lemma:pauli-mixing-to-2design}
  Let $\cE$ be an ensemble of Clifford unitaries and
  $\cE_\cQ$ be defined as above. If $\cE$ is Pauli mixing,
  then $\cE_\cQ$ implements the full bilateral twirl.  
\end{lemma}
The original proof of Lemma \ref{lemma:pauli-mixing-to-2design}
can be found in \cite{DLT02}.
A short proof based on representation theory can be found in
\cite{GAE2007}.
In Appendix \ref{appendix:pmi2d} we provide an elementary proof that
may be of independent interest. This proof uses some ideas from
\cite{DLT02} but has fewer assumptions. In particular, the new proof
does not rely on knowing how to evaluate (in closed form) the full
bilateral twirl of an arbitrary input state, nor on knowing the
invariants of the full bilateral twirl.  It is known how to evaluate
the full bilateral twirl using representation theory or the double
commutant theorem.  Our new proof derives this result (how the 
full bilateral twirl acts) on the side.

Note that, in light of Lemma~\ref{lem:equivdef}, an alternative way of proving that, whenever $\cE$ is Pauli mixing, $\cE_\cQ$ is a unitary $2$-design is to 
use Definition~\ref{def:channeltwirl}.
This can be shown in the following two steps.
First, conjugating any channel by a uniformly random Pauli operation drawn from $\cQ_n$ yields a \textit{mixed-Pauli channel} (a channel that is a probability distribution on the Pauli operators).
This is proved in~\cite{DCEL09}.
Second, it is clear that, if $\cE$ is Pauli mixing, then conjugating any mixed Pauli channel by a random element of $\cE$ results in a depolarizing channel with the same average fidelity as the mixed Pauli channel.
This corresponds exactly to an implementation of the full channel twirl.

\section{Pauli mixing using the structure of {\rm $\SL_2(\GF(2^n))$}}

For the purposes of analyzing the Clifford group and its action on the
Pauli group elements $X^aZ^b = (X^{a_1} \otimes \cdots \otimes
X^{a_n})(Z^{b_1} \otimes \cdots \otimes Z^{b_n})$, it is fruitful to
associate $a$ and $b$ with elements of the Galois field of size $2^n$.
However, for this association to work well technically, we work with
two different representations of field elements.  If $a$ is
represented in some \textit{primal} basis then $b$ is represented in
the \textit{dual} of that basis.  This section explains this basic
framework.

\subsection{Review of some properties of Galois fields {\rm $\GF(2^n)$}}
\label{sec:field-basics}

Let $\GF(2^n)$ denote the Galois field of size $2^n$ (more information about these fields can be found in~\cite{mLID94a}).
The elements of this field form a vector space over $\GF(2)$ so the notion of a basis is well-defined: 
$\omega_1, \dots, \omega_n \in \GF(2^n)$ are a \textit{basis} if they are linearly independent and span the field; a basis enables us to associate the elements of $\GF(2^n)$ with $n$-bit strings.
A \textit{polynomial basis} of $\GF(2^n)$ is a basis that is of the
form $1, \alpha, \alpha^2, \dots, \alpha^{n-1}$, for some 
$\alpha \in \GF(2^n)$.  The standard constructions of $\GF(2^n)$ in
terms of irreducible polynomials result in a polynomial basis.
However, there are bases that are not necessarily of this form, and these arise in our constructions.  
For instance, a \textit{normal basis} of $\GF(2^n)$ has the form
$\alpha^{2^0}, \alpha^{2^1}, \dots, \alpha^{2^{n{-}1}}$ for some
$\alpha \in \GF(2^n)$.    

The dual of a basis is defined in terms of the \textit{trace} function $T : \GF(2^n) \rightarrow \GF(2)$, which is defined as 
$T(a) = a^{2^{0}} + a^{2^{1}} + \cdots + a^{2^{n-1}}$.
The trace has the property that $T(a+b) = T(a) + T(b)$, for all $a, b \in \GF(2^n)$.
In terms of $T$, we can define the \textit{trace inner product} of $a, b \in \GF(2^n)$ as $T(ab)$.
Now, for an arbitrary basis $\omega_1, \dots, \omega_n \in \GF(2^n)$, that we refer to as the \textit{primal basis}, we can define its \textit{dual basis} as the unique 
$\hat{\omega}_1, \dots, \hat{\omega}_n \in \GF(2^n)$ such that 
\begin{align}\label{eq:def-self-dual}
T(\omega_i \hat{\omega}_j) &= \left\{
\begin{array}{cl}
1 & \mbox{if $i=j$} \\
0 & \mbox{if $i\neq j$.}
\end{array}
\right.
\end{align}
We can associate the elements of $\GF(2^n)$ with $n$-bit binary strings  by taking coordinates with respect to a basis.
To facilitate discussion, we use the following notation.
With respect to any primal basis $\omega_1, \dots, \omega_n$ and its dual $\hat{\omega}_1, \dots, \hat{\omega}_n$, for $a \in\GF(2^n)$: 
\begin{itemize}
\item $\lceil a \rceil \in \{0,1\}^n$ denotes the coordinates of $a$ in the primal basis.
Thus, 
$a = \lceil a \rceil_1\, \omega_1 + \cdots + \lceil a \rceil_n\, \omega_n$, which is achieved by setting $\lceil a \rceil_j = T(a\hat{\omega}_j)$ for all $j\in\{1,\cdots, n\}$.
\item $\lfloor a \rfloor \in \{0,1\}^n$ denotes the coordinates of $a$ in the dual basis.
Thus, 
$a = \lfloor a \rfloor_1\, \hat{\omega}_1 + \cdots + \lfloor a \rfloor_n\, \hat{\omega}_n$, which is achieved by setting $\lfloor a \rfloor_j = T(a\omega_j)$ for all $j\in \{1,\cdots, n\}$.
\end{itemize}
In some places, where the meaning is clear from the context, it is convenient to write $a$ in place of $\lceil a \rceil$.
Also, it is sometimes convenient to think of $n$-bit binary strings as $\{0,1\}$-valued column vectors of length~$n$.
Thus, $\lceil a \rceil$ and $\lfloor a \rfloor$ are sometimes interpreted as binary column vectors of length~$n$.
Binary matrices acting on these vectors (in mod 2 arithmetic) are written with square brackets.  

It is straightforward to show that the conversion from primal to dual basis  coordinates corresponds to multiplication by the $n \times n$ binary matrix 
\begin{align}\label{basis-convert}
W =
\begin{bmatrix}
T(\omega_1\omega_1) & \cdots & T(\omega_1\omega_n) \\
\vdots & \ddots & \vdots \\
T(\omega_n\omega_1) & \cdots & T(\omega_n\omega_n) 
\end{bmatrix}.
\end{align}
That is, $\lfloor a \rfloor = W \lceil a \rceil$ (with matrix-vector multiplication 
in mod 2 arithmetic).  
Also, $T(ab)$ is the dot-product of the coordinates of $a$ in the primal basis and the coordinates of $b$ in the dual basis: 
\begin{align}
T(ab) = \lceil a \rceil \cdot \lfloor b \rfloor = 
\lceil a\rceil_1 \lfloor b \rfloor_1 + \cdots + \lceil a\rceil_n \lfloor b \rfloor_n \bmod 2.
\end{align}
The dual of the dual basis is the primal basis. A basis is \textit{self-dual} if $\omega_i = \hat{\omega}_i$ for all $i$.

Relative to any basis, multiplication by any particular $r \in \GF(2^n)$ is a linear operator in the following sense.
There exists a binary $n \times n$ matrix $M_r$ such that, for all $s \in \GF(2^n)$, 
$\lceil rs \rceil = M_r \lceil s \rceil$ (with mod 2 arithmetic for the matrix-vector multiplication).
In fact, this matrix $M_r$ is 
\begin{align}\label{matrix-M_r}
M_r =
\begin{bmatrix}
T(r\,\hat{\omega}_1 \omega_1) & \cdots & T(r\,\hat{\omega}_1 \omega_n) \\
\vdots & \ddots & \vdots \\
T(r\,\hat{\omega}_n \omega_1) & \cdots & T(r\,\hat{\omega}_n \omega_n)
\end{bmatrix},
\end{align}
and its transpose $(M_r)^{\mathsf{T}}$ corresponds to multiplication by $r$ in the dual basis (that is, $\lfloor rs \rfloor = (M_r)^{\mathsf{T}}\lfloor s \rfloor$).
It should be noted that algorithms for multiplication in $\GF(2^n)$ are basis dependent; the obvious cost of converting between two bases is $O(n^2)$.

\subsection{Pauli mixing from a subgroup isomorphic to {\rm $\SL_2(\GF(2^n))$}}

Due to Lemma \ref{lemma:pauli-mixing-to-2design}, it suffices to
compute an ensemble of Clifford unitaries that is Pauli mixing.
Relative to a (primal) basis, we associate each pair $a, b \in \GF(2^n)$ with the Pauli group element 
$X^{\lceil a \rceil} Z^{\lfloor b \rfloor} 
= (X^{\lceil a \rceil_1} \otimes \cdots \otimes X^{\lceil a \rceil_n})(Z^{\lfloor b \rfloor_1} \otimes \cdots \otimes Z^{\lfloor b\rfloor_n})$.
Chau~\cite{Chau2005} showed that there is a subgroup $\mathscr{C}$ of
the Clifford group of size $2^{O(n)}$ such that sampling uniformly
over $\mathscr{C}$ performs Pauli mixing.  We now give an
overview of the approach in~\cite{Chau2005} (translated into our language).  
The subgroup
$\mathscr{C}$ is isomorphic to the special linear group of $2 \times
2$ matrices over $\GF(2^n)$:
\begin{align*}
  \SL_2(\GF(2^n)) = \biggl\{
	\begin{pmatrix}
	  \alpha & \beta \\
	  \gamma & \delta
	\end{pmatrix} : 
	\mbox{$\alpha, \beta, \gamma, \delta \in \GF(2^n)$ such that 
	$\alpha\delta + \beta\gamma = 1$}
  \biggr\}.
\end{align*}
Note that $\SL_2(\GF(2^n))$ has $2^{3n}-2^n$ elements. 
The subgroup $\mathscr{C}$ induces a group action of $\SL_2(\GF(2^n))$
on the Paulis by conjugation by certain unitaries.
\begin{defn}\label{def:induce}
With respect to a primal basis for $\GF(2^n)$, we say that a Clifford 
unitary
$U$ \emph{induces} $M\in\SL_2(\GF(2^n))$ if, for all $a, b \in
\GF(2^n)$ and
\begin{align}\label{eq:multiply-by-U}
\begin{pmatrix}
a' \\
b'
\end{pmatrix}
= M
\begin{pmatrix}
a \\
b
\end{pmatrix},
\end{align}
\begin{align}\label{eq:conjugate-by-U}
U X^{\lceil a \rceil} Z^{\lfloor b \rfloor} U^{\dagger} 
\equiv X^{\lceil a' \rceil} Z^{\lfloor b' \rfloor},
\end{align}
where $\equiv$ means equal up to a global phase in $\{1,i,-1,-i\}$ that 
is a function of $M\!$, $a$, and $b$.  
\end{defn}

\noindent To rephrase the above definition, suppose 
$M = \bigl(\begin{smallmatrix} \alpha & \beta \\ \gamma & \delta \end{smallmatrix}\bigr) \in G_0$.
Then, for all $a,b$, 
conjugating $(X^{\lceil a \rceil_1}
\otimes \cdots \otimes X^{\lceil a \rceil_n})(Z^{\lfloor b \rfloor_1}
\otimes \cdots \otimes Z^{\lfloor b\rfloor_n})$
by the Clifford unitary $U$ yields 
$(X^{\lceil \alpha a + \beta b \rceil_1}
\otimes \cdots \otimes X^{\lceil \alpha a + \beta b\rceil_n})
(Z^{\lfloor \gamma a + \delta b \rfloor_1}\otimes \cdots 
\otimes Z^{\lfloor \gamma a + \delta b \rfloor_n})$ up to a phase.  

We adopt the following notational convention throughout the paper.  We write
matrices in $\SL_2(\GF(2^n))$ and vectors of length $2$ over
$\GF(2^n)$ using parenthesis (see above) to distinguish the binary
matrices and vectors described in the previous subsection which use square brackets.

It should be noted that, in~\cite{Chau2005},
Eq.~\eqref{eq:conjugate-by-U} is expressed using different notation
for the Paulis, that we call \textit{subscripted} Paulis, defined as
satisfying $X_a\ket{c} = \ket{a+c}$ and $Z_b\ket{c} =
(-1)^{T(bc)}\ket{c}$.
It is easy to express these in terms of our \textit{superscripted}
Paulis, $X^{\lceil a \rceil}$ and $Z^{\lfloor b \rfloor}$, as 
$X_a = X^{\lceil a \rceil}$ and $Z_b = Z^{\lfloor b \rfloor}$ (since
$T(bc) = \lfloor b \rfloor {\cdot} \lceil c\rceil$).
The occurence of the dual basis in $Z^{\lfloor b \rfloor}$ (which is equivalent
to using $Z_b$) in Eq.~\eqref{eq:conjugate-by-U} is not merely a
matter of convention: for general $M \in \SL_2(\GF(2^n))$ there
\textit{does not exist} a unitary $U$ that induces $M$ in the sense
that $U X^{\lceil a \rceil} Z^{\lceil b \rceil} U^{\dagger} 
\equiv X^{\lceil a' \rceil} Z^{\lceil b' \rceil}$.  In terms of
Definition~\ref{def:induce}, the following holds.

\begin{lemma}[\cite{Chau2005}]\label{lemma:induce}
With respect to any primal basis of $\GF(2^n)$ and every $M\in\SL_2(\GF(2^n))$, there exists an $n$-qubit Clifford unitary $U$ that induces $M$.
\end{lemma}

\begin{defn}
Consider $M\in\SL_2(\GF(2^n))$. 
Let $U_M$ denote a unitary that induces $M$ with respect to the primal basis; 
$U_M$ is unique up to multiplication by a Pauli (a proof of this can be found in
\cite{gottesman1997, gottesman1998}, and is also provided in 
Appendix~\ref{appendix:uniqueness}, Lemma~\ref{lemma:appendix-uniqueness}).
Similarly, let $\widehat{U}_M$ denote a unitary that induces $M$ with respect to the dual basis.
\end{defn}

The proof of Lemma~\ref{lemma:induce} in \cite{Chau2005} exhibits a possible choice of $U_M$ for any $M \in \SL_2(\GF(2^n))$.
However, it is unclear how to implement that $U_M$ as a small quantum circuit, except for the fact that $U_M$ is in the Clifford group, so its gate complexity is $O(n^2/\log n)$ by~\cite{AG2004}.
Our results in subsequent sections amount to an alternative proof of Lemma~\ref{lemma:induce} for certain bases of $\GF(2^n)$, as well as a modified version of this lemma.
This enables us to ultimately attain gate constructions of size $\widetilde{O}(n)$ that implement unitary 2-designs.
The relationship between Lemma~\ref{lemma:induce} and unitary 2-designs is based on the fact that the uniform ensemble over $\{U_M:M \in \SL_2(\GF(2^n))\}$ is Pauli mixing, which is a consequence of the following.

\begin{lemma}[\cite{Chau2005,GAE2007}]\label{lemma41}
  Let $G_0$ denote the set of all non-zero elements of $\GF(2^n)\times
  \GF(2^n)$. 
  Let $M\in \SL_2(\GF(2^n))$ be chosen uniformly at random. Then, for any 
  $\bigl(\begin{smallmatrix}
	  a \\ b
\end{smallmatrix}\bigr) \in G_0$, 
\be\begin{pmatrix} c\\d\end{pmatrix} = M\begin{pmatrix}a\\b\end{pmatrix}\ee
  is uniformly distributed over $G_0$.
\end{lemma}

\begin{proof}
We first show that $\SL_2(\GF(2^n))$ acts transitively on $G_0$.
Let
$\bigl(\begin{smallmatrix} c \\ d \end{smallmatrix}\bigr) \in G_0$.
If $c \neq 0$, 
then, 
$\bigl(\begin{smallmatrix} c & 0 \\ d & c^{{-}1} \end{smallmatrix}\bigr)
\bigl(\begin{smallmatrix} 1 \\ 0 \end{smallmatrix}\bigr) = 
\bigl(\begin{smallmatrix} c \\ d \end{smallmatrix}\bigr)$.  
If $c = 0$, then $d \neq 0$, so, 
$\bigl(\begin{smallmatrix} 0 & d^{-1} \\ d & 0 \end{smallmatrix}\bigr)
\bigl(\begin{smallmatrix} 1 \\ 0 \end{smallmatrix}\bigr) = 
\bigl(\begin{smallmatrix} c \\ d \end{smallmatrix}\bigr)$.  
Thus, we can map any 
$\bigl(\begin{smallmatrix} c_1 \\ d_1 \end{smallmatrix}\bigr) \in G_0$
to 
$\bigl(\begin{smallmatrix} 1 \\ 0 \end{smallmatrix}\bigr)$ 
and then to any other 
$\bigl(\begin{smallmatrix} c_2 \\ d_2 \end{smallmatrix}\bigr) \in G_0$ 
using elements of $\SL_2(\GF(2^n))$.  

To prove the lemma, suppose, by contradiction, that there are distinct 
$\bigl(\begin{smallmatrix} c_1 \\ d_1 \end{smallmatrix}\bigr)$, 
$\bigl(\begin{smallmatrix} c_2 \\ d_2 \end{smallmatrix}\bigr)$ such that 
${\rm Prob}_M \left\{
M \bigl(\begin{smallmatrix} a \\ b \end{smallmatrix}\bigr)
 = \bigl(\begin{smallmatrix}
	  c_1 \\ d_1 \end{smallmatrix}\bigr) \right\} = p_1$, 
${\rm Prob}_M \left\{
M \bigl(\begin{smallmatrix} a \\ b \end{smallmatrix}\bigr)
 = \bigl(\begin{smallmatrix}
	  c_2 \\ d_2 \end{smallmatrix}\bigr) \right\} = p_2$
and $p_1 > p_2$.  
But there exists an $M' \in  \SL_2(\GF(2^n))$ such that 
$M' \bigl(\begin{smallmatrix} c_1 \\ d_1 \end{smallmatrix}\bigr)
= 
\bigl(\begin{smallmatrix} c_2 \\ d_2 \end{smallmatrix}\bigr)$. 
So, 
${\rm Prob}_M \left\{
M' M \bigl(\begin{smallmatrix} a \\ b \end{smallmatrix}\bigr)
 = \bigl(\begin{smallmatrix}
	  c_2 \\ d_2 \end{smallmatrix}\bigr) \right\} \geq p_1$. 
But the distribution over $M$ is the same as the distribution 
over $M' M$, so the left side of the last inequality is $p_2$, 
which is a contradiction. 
\end{proof}

Our goal is to implement $U_M$ (for $M \in \SL_2(\GF(2^n))$) with quantum circuits consisting of $\widetilde{O}(n)$ Clifford gates.
The interplay between the primal basis and the dual basis is a major complicating factor that we address using two different approaches.
In one of our approaches we modify the framework of $\SL_2(\GF(2^n))$.

Our approach in Section~\ref{sec:mult} is based on a self-dual basis for $\GF(2^n)$ and the structure of $\SL_2(\GF(2^n))$.
Our approach in Section~\ref{sec:mult-poly} is based on a polynomial basis for $\GF(2^n)$ (and its dual) and the structure of two subgroups of $\SL_2(\GF(2^n))$: the lower triangular subgroup and the upper triangular subgroup. These are defined respectively as 
\begin{align}
  \LT_2(\GF(2^n)) &= \biggl\{
	\begin{pmatrix}
	  \alpha & 0 \\
	  \beta & \alpha^{-1}
	\end{pmatrix} : 
	\mbox{$\alpha, \beta \in \GF(2^n)$ and $\alpha \neq 0$}
  \biggr\} \\[1mm]
  \UT_2(\GF(2^n)) &= \biggl\{
	\begin{pmatrix}
	  \alpha & \beta \\
	  0 & \alpha^{-1}
	\end{pmatrix} : 
	\mbox{$\alpha, \beta \in \GF(2^n)$ and $\alpha \neq 0$}
  \biggr\}.
\end{align}
These subgroups have interesting mixing properties, albeit weaker ones than 
$\SL_2(\GF(2^n))$, which are explained in Section~\ref{sec:mult-poly}.

\subsection{A framework for implementing elements of $\SL_2(\GF(2^n))$ by unitaries}\label{sec:framework}




We first show that all elements of $\SL_2(\GF(2^n))$ can be written as
a product of a small constant number of matrices in a generating
set---and more restrictive generating sets for $\LT_2(\GF(2^n))$ and
$\UT_2(\GF(2^n))$.
Then we describe Clifford unitaries that induce these generating matrices.
In subsequent sections, we show how to implement these unitaries 
with $\widetilde{O}(n)$ quantum gates, thereby implementing elements of 
$\SL_2(\GF(2^n))$.


\begin{lemma}\label{lemma:gaussian_elim}
  Every element $M \in \SL_2(\GF(2^n))$ can be expressed as a product of a
  constant number of the following elements of $\SL_2(\GF(2^n))$:
\begin{align}\label{eq:elementary-SL}
\begin{pmatrix}
	  r & 0\\ 0 & r^{-1}
	\end{pmatrix},\ \ 
	\begin{pmatrix}
	  1 & 0 \\ 1 & 1
	\end{pmatrix},\ \ 
	\begin{pmatrix}
	  0 & 1\\ 1 & 0
  \end{pmatrix}
\end{align}
where $r\in \GF(2^n)$ is non-zero.
\end{lemma}
\begin{proof}
For any $M = \begin{pmatrix}
\alpha & \gamma\\ \beta& \delta\end{pmatrix} \in \SL_2(\GF(2^n))$, we can
decompose it into a product as follows:
\begin{align}\label{eq:decomposition}
\begin{pmatrix} \alpha & \gamma \\ \beta & \delta
\end{pmatrix} = 
\begin{cases}
\begin{pmatrix} 1 & 0\\ \frac\beta\alpha&
1\end{pmatrix}
\begin{pmatrix}1&\alpha\gamma\\0&1\end{pmatrix}\begin{pmatrix}\alpha&0\\0&\alpha^{-1}\end{pmatrix} 
& \mbox{if $\alpha \neq 0$} \vspace*{2mm}\\
\begin{pmatrix}
1 & 0\\ 
\frac{\delta}{\gamma} & 1
\end{pmatrix}
\begin{pmatrix}
\gamma & 0 \\
0 & \gamma^{-1}
\end{pmatrix}
\begin{pmatrix}
0 & 1 \\
1 & 0
\end{pmatrix} 
& \mbox{if $\alpha = 0$.} 
\end{cases}
\end{align}
Furthermore, for any non-zero $s \in \GF(2^n)$, there exists $t \in \GF(2^n)$ such that $t^2 = s$ (explicitly $t = s^{2^{n-1}}$).
This permits us to decompose further as 
\begin{align}
\begin{pmatrix} 1 & 0\\ s & 1\end{pmatrix} =
\begin{pmatrix}t^{-1}&0\\0& t\end{pmatrix}
\begin{pmatrix}1&0\\1&1\end{pmatrix}
\begin{pmatrix}t&0\\0&t^{-1}\end{pmatrix}
\end{align}
and
\begin{align}
\begin{pmatrix} 1 & s\\ 0 & 1\end{pmatrix} =
\begin{pmatrix}0&1\\1&0\end{pmatrix}
\begin{pmatrix} 1 & 0\\ s & 1\end{pmatrix}
\begin{pmatrix}0&1\\1&0\end{pmatrix}.
\end{align}
\end{proof}

It is easy to specialize the above lemma to the lower triangular and upper triangular matrices in $\SL_2(\GF(2^n))$ as follows. 

\begin{lemma}\label{lemma:gaussian_elim-triangular}
Every element of $\LT_2(\GF(2^n))$ can be expressed as a product of a
constant number of elements of the form 
\begin{align}\label{eq:elementary-LT}
\begin{pmatrix}
	  r & 0\\ 0 & r^{-1}
	\end{pmatrix}\ \ \mbox{and} \ \ 
	\begin{pmatrix}
	  1 & 0 \\ 1 & 1
	\end{pmatrix}
\end{align}
and every element $\UT_2(\GF(2^n))$ can be expressed as a product of a
constant number of elements of the form
\begin{align}\label{eq:elementary-UT}
\begin{pmatrix}
	  r & 0\\ 0 & r^{-1}
	\end{pmatrix}\ \ \mbox{and} \ \ 
	\begin{pmatrix}
	  1 & 1 \\ 0 & 1
	\end{pmatrix},
\end{align}
where $r\in \GF(2^n)$ is non-zero.
\end{lemma}


In view of Lemma~\ref{lemma:gaussian_elim}, for every $M$ in
$\SL_2(\GF(2^n))$, we can find a unitary that induces $M$ if
we find a unitary that induces each of 
$\bigl(\begin{smallmatrix}
	  1 & 0\\ 1 & 1
\end{smallmatrix}\bigr)$, 
$\bigl(\begin{smallmatrix}
	  0 & 1\\ 1 & 0
\end{smallmatrix}\bigr)$, 
and 
$\bigl(\begin{smallmatrix}
	  r & 0\\ 0 & r^{-1}
\end{smallmatrix}\bigr)$ for any non-zero $r \in \GF(2^n)$.  
Similar statements hold for $\LT_2(\GF(2^n))$ and 
$\UT_2(\GF(2^n))$ with their respectively generating sets 
shown in Lemma \ref{lemma:gaussian_elim-triangular}.  


First consider any non-zero $r \in \GF(2^n)$, and the element of $\SL_2(\GF(2^n))$ 
of the form   
\begin{align}
\begin{pmatrix}
	  r & 0\\ 0 & r^{-1}
	\end{pmatrix}. 
 \end{align}
A Clifford unitary that induces 
$\bigl(\begin{smallmatrix}
 	  r & 0\\ 0 & r^{-1}
\end{smallmatrix}\bigr)$ 
is the \textit{multiply-by-$r$} (in the primal basis) operation $\Pi_r$ defined%
\footnote{This unitary operation acts on computation basis states 
similarly to $M_r$ defined in 
Sec.~\ref{sec:field-basics}, Eq.~\eqref{matrix-M_r}.}
as $\Pi_r\ket{\lceil c \rceil} = \ket{ \lceil rc \rceil}$.
To improve readability, we henceforth
denote $\ket{\lceil c \rceil}$  by $\ket{c}$.  
For example, in this notation, $\Pi_r\ket{c} = \ket{rc}$.   
Now, to check that $\Pi_r$ induces 
$\bigl(\begin{smallmatrix}
 	  r & 0\\ 0 & r^{-1}
\end{smallmatrix}\bigr)$, note that, for all $c \in \GF_2(2^n)$, 
\begin{align}
\Pi_r X^{\lceil a \rceil} \Pi_{r}^{\dag} \, \ket{c} 
&= \Pi_r X^{\lceil a \rceil} \ket{r^{-1}c} \\
&= \Pi_r \ket{r^{-1}c+a} \\
&= \ket{c+ra} \\
&= X^{\lceil ra \rceil}\ket{c} \,.
\end{align}
Furthermore, 
\begin{align}
\Pi_r Z^{\lfloor b \rfloor} \Pi_{r}^{\dag} \, \ket{c} 
&= \Pi_r Z^{\lfloor b \rfloor} \ket{r^{-1}c} \\
&= \Pi_r (-1)^{\lfloor b \rfloor \cdot \lceil r^{-1}c \rceil} \ket{r^{-1}c} \\
&= (-1)^{T(br^{-1}c)} \ket{c} \label{eq:SL2_trace_1}\\
&= (-1)^{\lfloor br^{-1} \rfloor \cdot \lceil c \rceil} \ket{c} 
\label{eq:SL2_trace_2}\\
&= Z^{\lfloor br^{-1} \rfloor} \, \ket{c}.
\end{align}
It follows that, for all $a, b \in \GF(2^n)$, $\Pi_r X^{\lceil a \rceil} Z^{\lfloor b \rfloor} \Pi_{r}^{\dag} = X^{\lceil ra \rceil} Z^{\lfloor r^{-1}b \rfloor}$. In other words, $\Pi_r$ induces $\bigl(\begin{smallmatrix} r & 0\\ 0 & r^{-1}\end{smallmatrix}\bigr)$.

We can write any $\bigl(\begin{smallmatrix}a\\b\end{smallmatrix}\bigr)\in \SL_2(\GF(2^n))$ in a primal-dual basis as
$\left[ \begin{smallmatrix} \lceil a \rceil \\ \lfloor b \rfloor \end{smallmatrix}\right]\in\{0, 1\}^{2n}$, where $\lceil a \rceil, \lfloor b \rfloor \in\{0, 1\}^n$. 
Recall that to distinguish elements of $\SL_2(\GF(2^n))$ from their corresponding binary vectors in a primal-dual basis, we use parenthesis to denote the former and square brackets for the binary vectors and their linear operators. 

We summarize the effect of conjugating a Pauli $X^{\lceil a \rceil}
Z^{\lfloor b \rfloor}$ by $\Pi_r$ on the binary strings $\lceil a \rceil$ and
$\lfloor b \rfloor$ as the following mapping on $2n$-bit strings:
\begin{align}
\label{eq:bmapping-m-by-r}
\begin{bmatrix} 
\lceil a \rceil \vspace*{1.5mm}\\
\lfloor b \rfloor
\end{bmatrix}
\mapsto\
\begin{bmatrix} 
\lceil ra \rceil \vspace*{1.5mm}\\
\lfloor r^{-1}b \rfloor
\end{bmatrix}
=
\begin{bmatrix}
M_r \, \lceil a \rceil \vspace*{1.5mm}\\
(M_{r^{-1}})^{\sf T} \, \lfloor b \rfloor
\end{bmatrix}.
\end{align}
Here $M_r$ is the linear operator corresponding to multiplication by $r$ in the primal basis, as defined in Eq.~\eqref{matrix-M_r}; 
$(M_{r^{-1}})^{\sf T}$, the transpose of $M_{r^{-1}}$, which (due to the form of Eq.~\eqref{matrix-M_r}) is the linear operator corresponding to multiplication by $r^{-1}$ in the dual basis.

The following definition is similar to Definition~\ref{def:induce}. 
\begin{defn}
\label{def:induce2}
We say that a Clifford 
unitary
$U$ \emph{induces} the $2n \times 2n$ binary matrix $\mathscr{M}$, 
if, for all $n$-bit 
strings $s,r$, 
\begin{align}
U X^r Z^s U^{\dagger} 
\equiv X^{r'} Z^{s'},~~ {\rm where} ~~
\begin{bmatrix}
r' \vspace*{1mm}\\ s'
\end{bmatrix}
= 
\mathscr{M} 
\begin{bmatrix}
r \vspace*{1mm}\\ s
\end{bmatrix}.
\end{align}
Here, $\equiv$ means equal up to a global phase in $\{1,i,-1,-i\}$ that 
is a function of $\mathscr{M}$, $r$, and $s$.  
We also say that $U$ induces the mapping  $v \mapsto \mathscr{M} v$.  
\end{defn}

\noindent For example, 
$\Pi_r$ induces the mapping given by Eq.~(\ref{eq:bmapping-m-by-r})
and the matrix 
$\begin{bmatrix}
M_r & 0 \vspace*{1mm}\\ 0 & (M_{r^{-1}})^{\sf T}
\end{bmatrix}.
$

Now we return to finding unitaries that induce elements of $\SL_2(\GF(2^n))$.   
Consider the element 
$\bigl(\begin{smallmatrix}
	  1 & 0\\ 1 & 1
\end{smallmatrix}\bigr)$ 
of $\SL_2(\GF(2^n))$.
The Clifford unitary that induces 
$\bigl(\begin{smallmatrix}
	  1 & 0\\ 1 & 1
\end{smallmatrix}\bigr)$ 
should transform the Pauli $X^{\lceil a \rceil} Z^{\lfloor b \rfloor}$ 
along the lines of the mapping
\begin{align}
\begin{bmatrix} 
\lceil a \rceil \vspace*{1.5mm}\\
\lfloor b \rfloor
\end{bmatrix}
\mapsto\
\begin{bmatrix} 
\lceil a \rceil \vspace*{1.5mm}\\
\lfloor a+b \rfloor
\end{bmatrix}
=
\begin{bmatrix} 
\lceil a \rceil \vspace*{1.5mm}\\
\lfloor a \rfloor + \lfloor b \rfloor
\end{bmatrix}
=
\begin{bmatrix}
\lceil a \rceil \vspace*{1.5mm}\\ W \lceil a \rceil + \lfloor b \rfloor
\end{bmatrix}
=
\begin{bmatrix}
I & 0 \vspace*{1.5mm}\\ W & I
\end{bmatrix}
\begin{bmatrix} 
\lceil a \rceil \vspace*{1.5mm}\\
\lfloor b \rfloor
\end{bmatrix}
\label{eq:add-b-to-a}
\end{align}
where $W$ is the linear operator for primal-to-dual basis conversion
defined in Eq.~\eqref{basis-convert}.  

For \textit{any} symmetric $n \times n$ binary matrix $V$, we show
that there is a diagonal Clifford unitary $\Gamma_V$ that implements
$\begin{bmatrix} I & \!\! 0 \vspace*{0.1mm}\\ V & \!\! I \end{bmatrix}$.  
The unitary $\Gamma_V$ is defined as 
\begin{align}\label{eq:Gamma_V}
\Gamma_V \ket{c} 
&=
\mbox{\Large $i$}^{\,\sum_{j=1}^{n} \sum_{k=1}^{n} V_{jk}c_jc_k}\ket{c}.
\end{align}
We begin with some preliminary observations. Since, for all $i,j \in
\{1,\dots,n\}$, $V_{ij} = V_{ji}$, an equivalent definition is
\begin{align}\label{eq:Gamma_V-alternate}
\Gamma_V \ket{c} 
&=
\mbox{\Large $i$}^{\,\sum_{j=1}^{n} V_{jj}c_j}
(-1)^{\,\sum_{1 \leq j < k \leq n} V_{jk}c_jc_k}\ket{c}.
\end{align}
From Eq.~(\ref{eq:Gamma_V-alternate}), it is clear that $\Gamma_V$ is in the Clifford group, since it is computed by the following composition of gates: an $S$ gate acting on each qubit $j$ for which $V_{jj} = 1$; and a controlled-$Z$ gate acting on qubits $j$ and $k$ for each $j < k$ where $V_{jk} = 1$ (all these gates commute).
This generic construction consists of $O(n^2)$ gates.
In Sections~\ref{sec:mult} and~\ref{sec:mult-poly}, for the
primal-to-dual basis conversion matrix $W$ (which is symmetric from
its definition in Eq.~\eqref{basis-convert}), we exhibit circuits
implementing $\Gamma_W$ with $\widetilde{O}(n)$ gates.

To check that $\Gamma_W$ induces the mapping in Eq.~\eqref{eq:add-b-to-a},
it is convenient to separately consider the diagonal and off-diagonal entries of  $W$.
Let $W = D + E$, where $D$ is diagonal and $E_{jj} = 0$ for all $j \in \{1,\dots,n\}$. This allows us to write $\Gamma_W = \Gamma_{D + E} = \Gamma_D \Gamma_E$, as a direct consequence of 
\begin{align}
\begin{bmatrix}
	  I & 0\\ W & I
\end{bmatrix}
=
\begin{bmatrix}
	  I & 0\\ D & I
\end{bmatrix}
\begin{bmatrix}
	  I & 0\\ E & I
\end{bmatrix}.
\end{align}
Then, from the discussion following Eq.~(\ref{eq:Gamma_V-alternate}),
we know that $\Gamma_D = S^{W_{11}} \otimes \cdots \otimes S^{W_{nn}}$
and it is straightforward to check that
\begin{align}\label{eq:analyze-Gamma_D}
\Gamma_D X^{\lceil a \rceil} \Gamma_{D}^{\dag} 
= \left(\mbox{\Large $i$}^{W_{11}a_1 + \cdots + W_{nn}a_n}\right)
   X^{\lceil a \rceil} Z^{D \lceil a \rceil},
\end{align}
where we are using the notation $a_i = \lceil a \rceil_i$.
For $\Gamma_E$, we have
\begin{align}
\Gamma_E X^{\lceil a \rceil} \Gamma_{E}^{\dag}\ket{c} 
&= \Gamma_E X^{\lceil a \rceil} 
(-1)^{\,\sum_{j=1}^{n} \sum_{k=j+1}^{n} W_{jk}c_jc_k}\ket{c} \\
&= \Gamma_E
(-1)^{\,\sum_{j=1}^{n} \sum_{k=j+1}^{n} W_{jk}c_jc_k}\ket{a+c} \\
&= (-1)^{\,\sum_{j=1}^{n}\sum_{k=j+1}^{n} 
W_{jk}((a_j+c_j)(a_k+c_k) +c_jc_k)}\ket{a+c} \\
&= (-1)^{\,\sum_{j=1}^{n}\sum_{k=j+1}^{n} 
W_{jk}(a_j a_k + a_jc_k + a_kc_j)}\ket{a+c} \\
&= (-1)^{\sum_{j=1}^n \sum_{k=j+1}^n W_{jk} a_j a_k} 
(-1)^{ \lceil c \rceil \cdot  E \lceil a \rceil} |a+c\rangle \\
&= (-1)^{\sum_{j=1}^n \sum_{k=j+1}^n W_{jk} a_j a_k} X^{\lceil a \rceil} 
(-1)^{ \lceil c \rceil \cdot  E \lceil a \rceil} |c\rangle \\
&= (-1)^{\,\sum_{j=1}^{n}\sum_{k=j+1}^{n} 
W_{jk}a_j a_k} X^{\lceil a \rceil} Z^{E \lceil a \rceil}\ket{c} \,.
\label{eq:analyze-Gamma_E}
\end{align}
Combining Eqs.~\eqref{eq:analyze-Gamma_D}, \eqref{eq:analyze-Gamma_E}, 
and the fact that $\Gamma_W$ commutes with every $Z^{\lfloor b \rfloor}$, 
we have
\begin{align}
\Gamma_W X^{\lceil a \rceil} Z^{\lfloor b \rfloor} \Gamma_{W}^{\dag} & = 
\left(\mbox{\Large $i$}^{\,\sum_{j=1}^{n}\sum_{k=1}^{n}W_{jk} a_j a_k}\right)
X^{\lceil a \rceil} Z^{{D \lceil a \rceil} + E {\lceil a \rceil} + \lfloor b \rfloor}
\\
& = \left(\mbox{\Large $i$}^{\,\sum_{j=1}^{n}\sum_{k=1}^{n}W_{jk} a_j a_k}\right)
X^{\lceil a \rceil} Z^{{W \lceil a \rceil} + \lfloor b \rfloor}
,
\end{align}
which implies that $\Gamma_W$ induces 
the mapping in Eq.~\eqref{eq:add-b-to-a}.

For completeness, a unitary operation that induces the element
$\bigl(\begin{smallmatrix}
	  0 & 1\\ 1 & 0
\end{smallmatrix}\bigr)$
of $\SL_2(\GF(2^n))$ should also be considered.
This is addressed in sections~\ref{sec:mult} and~\ref{sec:mult-poly} in very different ways, and we defer the discussion of this to those sections.



\section{$\widetilde{O}(n)$ implementation based on self-dual basis for $\GF(2^n)$}\label{sec:mult}

We want to find $\widetilde{O}(n)$-sized circuits to implement
unitaries that induce the generators of $\SL_2(\GF(2^n))$.  Our first
approach is to represent $\GF(2^n)$ in a self-dual basis.
The advantage of using a self-dual basis is that, the change 
of basis operation $W$ defined in Eq.~(\ref{basis-convert}) is 
simply $I$.  
Since there is no distinction between coordinates in the primal and
the dual bases, we omit the $\lceil ~\rceil$ and $\lfloor
~\rfloor$ notations in this section.
For all $n$-bit strings $a, b$,  
$S^{\otimes n}X^a Z^{b}(S^{\dag})^{\otimes n} 
= i^{a_1 + \cdots + a_n \bmod 4} X^a Z^{a+b}$ 
and 
$H^{\otimes n}X^a Z^{b}H^{\otimes n} = (-1)^{a\cdot b} X^b Z^a$.
Therefore, $S^{\otimes n}$ and $H^{\otimes n}$ respectively induce 
$\bigl(\begin{smallmatrix}
	  1 & 0\\ 1 & 1
\end{smallmatrix}\bigr)$
and
$\bigl(\begin{smallmatrix}
	  0 & 1\\ 1 & 0
\end{smallmatrix}\bigr)$.  

The challenge of using a self-dual basis lies in the implementation of
the unitary $\Pi_r$ (field multiplication) that induces
$\bigl(\begin{smallmatrix}
	  r & 0\\ 0 & r^{-1}
\end{smallmatrix}\bigr)$.  
Fast multiplication methods with respect to a polynomial basis
are known; however, no polynomial basis of $\mathrm{GF}(2^n)$ is
also self-dual if $n\geq2$ \cite{hazewinkel1996}.
Our solution is to use special self-dual bases that can be efficiently
converted to and from polynomial bases.
These special self-dual bases are constructed with {\em
  Gauss periods}, and are known for {\em admissible} $n$'s (see
Definition \ref{def:adn} below).  According to \cite{von2001}, there
are {\em infinitely many} admissible $n$'s under the extended Riemmann
Hypothesis.  Our implementation in this section is 
restricted to these values of $n$: 
\begin{defn} \label{def:adn}
A natural number $n$ is called {\em admissible} if the following 
two conditions hold: 
\begin{itemize}
  \setlength\itemsep{0em}
  \item[(1)] $2n+1$ is prime
  \item[(2)] $\gcd(e,n) = 1$, where $e$ is the index of the subgroup generated by $2$ in $\mathbb Z_{2n+1}^*$.
\end{itemize}
In the above, $\mathbb Z_{2n+1}^*$ denotes the multiplicative group of 
$\mathbb Z_{2n+1}$. Since $\mathbb{Z}_{2n{+}1}^*$ has $2n$ elements, $e = \frac{2n}{|\<2\>|}$.
\end{defn}
\vspace*{12pt}

In the remainder of this section, we first describe the procedure of
finding a self-dual basis using Gauss periods, and briefly explain the
efficient conversion between these two representations. Then we
describe the implementation of $\Pi_r$.

Since, for admissible values of $n$, $2n+1$ is prime, Fermat's Little 
Theorem implies $2^{2n}\equiv1 \bmod {2n+1}$.
So $2n+1$ divides $2^{2n}-1$, which implies that
there is a primitive $(2n+1)$-th root of unity
$\beta\in\mathrm{GF}(2^{2n})$. One way to get $\beta$ is the following. Let
$\xi$ be a generator of the multiplicative group of $\GF(2^{2n})$. Because 
$\xi^{2^{2n}-1} = 1$, we can take $\beta = \xi^{(2^{2n}-1)/(2n+1)}$.
Consider the set
\be
  \mathcal{S} = \{\beta+\beta^{-1}, \beta^2+\beta^{-2}, \ldots,
  \beta^{n}+\beta^{-n}\}.
\ee
We first show that $\mathcal{S}$ is a self-dual normal basis of
$\mathrm{GF}(2^n)$ over $\mathrm{GF}(2)$ (as defined in
Section~\ref{sec:field-basics}).  Then we show how to 
efficiently convert between $\mathcal{S}$ and a polynomial basis.  

First we show that for an admissible $n$, $2$ and $-1$ generate
$\mathbb{Z}_{2n{+}1}^*$ (i.e., $\left<2, -1 \right> =
\mathbb{Z}_{2n{+}1}^*$).  A proof is given in \cite{Gao2000}, and 
it can be rephrased as follows.  
Let $\gamma$ generate the cyclic group $\mathbb{Z}_{2n{+}1}^*$.
If $e$ is the index of $\left<2 \right>$ in $\mathbb{Z}_{2n{+}1}^*$,
then $2 = \gamma^e$.
%
%
Furthermore, $\gamma^n = -1$.  Since $\gcd(e,n) = 1$, there
are integers $k_1, k_2$ such that $1 = e k_1 + n k_2$ and therefore,
$\gamma \in \left<2, -1 \right>$, so 
$\mathbb{Z}_{2n{+}1}^* = \left<2, -1 \right>$.  

Our next step showing $\mathcal{S}$ is a self-dual basis follows from \cite{von2007}.
Since $\mathbb{Z}_{2n{+}1}^* = \left<2, -1 \right>$, it follows that 
\be
  \{2^0,-2^0,2^1,-2^1,\ldots,2^{n-1},-2^{n-1}\} \equiv \{1,-1,2,-2, \dots, 
  n,-n\}\ \  \bmod {2n+1}.
\ee 
and we can reorder the elements of $\mathcal{S}$ as 
\be \label{eq:set_reordering}
   \{\beta^{2^0}+\beta^{{-}2^0}, \beta^{2^1}+\beta^{{-}2^1}, 
    \ldots, \beta^{2^{n{-}1}}+\beta^{{-}2^{n{-}1}}\}.
\ee
The set in Eq.~\eqref{eq:set_reordering} as a subset of $\GF(2^n)$ is equal to $\{\alpha^{2^0},
\alpha^{2^1}, \cdots, \alpha^{2^{n{-}1}}\}$ where $\alpha = \beta +
\beta^{{-}1}$ is called a {\em Gauss period} of type $(n,2)$
over $\mathrm{GF}(2)$. It is easy to see that $\beta + \beta^{-1} \in
\GF(2^n)$, for one can verify that $(\beta + \beta^{-1})^{2^n} = \beta +
\beta^{-1}$.

Finally, we need to show that $\mathcal{S}$ is a basis.  We invoke
Theorem 3.1 in \cite{Gao2000} which implies that $\alpha$ is a normal
element in $\mathrm{GF}(2^n)$ (generating a normal basis as defined in Section~\ref{sec:field-basics}).  Then, from Corollary 3.5
in \cite{Gao2000}, any normal basis of Gauss period of type $(n,2)$
over $\mathrm{GF}(2)$ is self-dual when $n>2$, so, $\mathcal{S}$ is
self-dual, as claimed.

Next, we show how to efficiently convert between $\mathcal{S}$ and a
polynomial basis.  
We define a mapping from $\GF(2^n)$ to $\{0, 1\}^{n+1}$ as follows. If 
$a\in \GF(2^n)$, then $a' = [0, a_1,\cdots, a_n]^T$, where
$a = a_1(\beta+\beta^{-1}) + \ldots + a_n(\beta^n+\beta^{-n})$. In other words, 
$a'$ is the coordinate of $a$ with respect to the spanning set 
 $\{1, \beta+\beta^{-1}, \beta^2+\beta^{-2}, \ldots, \beta^{n}+\beta^{-n}\}$.
Including the element $1$ makes this spanning set not a basis, but significantly
simplifies the conversion between the following two spanning sets: 
\begin{align}
  \mathcal{S'} = & \{1, \beta+\beta^{-1}, \beta^2+\beta^{-2}, \ldots,
  \beta^{n}+\beta^{-n}\} \,,
\label{eq:sprime}
\\  \mathcal{T} = & 
\{1, \beta+\beta^{-1},(\beta+\beta^{-1})^2,\ldots,(\beta+\beta^{-1})^n\}.
\label{eq:t}
\end{align}
Notice that the set $\mathcal T$ arises from adding $1$ to a polynomial basis.  
We call $\mathcal{S'}$ a {\em self-dual spanning set} and $\mathcal{T}$ a {\em
polynomial spanning set}.  The fact that $\mathcal{T}$ is not a basis does not
affect how we represent a field element as a polynomial based on $\mathcal{T}$, 
i.e., $a = \sum_{i=0}^n a_i(\beta+\beta^{-1})^i$, and fast multiplication of two
polynomials of this form still works.  

Let $s_i = \beta^i + \beta^{-i}$, $t_i = (\beta + \beta^{-1})^i$, and let $s_i'$
and $t_i'$ be the $(n+1)$-bit string output by the mapping defined earlier.
We now describe the linear transformation $L_{n+1}$ 
that maps $s_i'$ to $t_i'$ for all $i$ (by right multiplication).  
The transformation $L_{n+1}$ is not unique.
A simple choice for $L_{n+1}$ is based on the binomial expansion 
$(\beta+\beta^{{-}1})^j = \sum_{i=0}^j {j \choose i} \beta^{j-2i}$. 
More precisely, for general $k$, we can choose $L_k$ as
\begin{align}\label{lmatrix}
  \left(L_{k}\right)_{i, j} = \left\{
	\begin{array}{l l}
	  0 & \quad \text{if $i > j$ or $j - i$ is odd,} \\
	  \binom{j}{(j-i)/2} \bmod 2 & \quad \text{otherwise,}
	\end{array} \right.
\end{align}
where $0\leq i,j < k$. The operation $L_k$ can be reversed. $L_k$ is
upper-triangular with 1's on the diagonal, which implies
$\det(L_k)=1$, so $L_k$ is invertible.

Finally, we will find a unitary $\mathscr{L}_k$ that induces 
$L_k$.  (More precisely, we are inducing the matrix with identical
diagonal blocks that is the $(k{-}1)\times(k{-}1)$ submatrix of $L_k$
with the first row and column omitted.)  The unitary $\mathscr{L}_n$ 
also induces a conversion from $\mathcal{S'}$ to $\mathcal{T}$.  
In \cite{von2007}, the following theorem is proved. 
\begin{theorem}[\cite{von2007}] \label{thm4132}
  Right multiplying $L_{n+1}$ ($L^{-1}_{n+1}$ respectively) by the
  vector representation ($a'$) of an element $a \in \mathrm{GF}(2^n)$
  described above can be done using $O(n\log n)$ operations (additions
  and multiplications) in $\mathrm{GF}(2)$.
\end{theorem}
\noindent From this theorem, an efficient (classical) circuit for 
$L_{n+1}$ can be built
with $O(n\log n)$ CNOT gates.  The intuition is that $L_{n+1}$ 
can be decomposed as a product of $O(\log n)$ matrices,
each with $O(n)$ 1's. Since the linear transformation can be done with
$\mathrm{GF}(2)$ additions and multiplications, it can be implemented
with CNOT gates.  A circuit for $L^{-1}_{n+1}$ can be obtained by 
running the circuit for $L_{n+1}$ backwards.  

Here we prove Theorem \ref{thm4132} with a different approach -- a
recursive construction that also requires $O(n\log n)$ CNOT gates.
First consider $L_k$ as defined in Eq.~\eqref{lmatrix} where $k = 2^t$ 
is a power of $2$.   
Taking $k=8$ as an example,
\begin{align}
L_8 =
\left[
\begin{array}{cccc|cccc}
1 & 0 & 0 & 0 & 0 & 0 & 0 & 0 \\
0 & 1 & 0 & 1 & 0 & 0 & 0 & 1 \\
0 & 0 & 1 & 0 & 0 & 0 & 1 & 0 \\
0 & 0 & 0 & 1 & 0 & 1 & 0 & 1 \\ \hline
0 & 0 & 0 & 0 & 1 & 0 & 0 & 0 \\
0 & 0 & 0 & 0 & 0 & 1 & 0 & 1 \\
0 & 0 & 0 & 0 & 0 & 0 & 1 & 0 \\
0 & 0 & 0 & 0 & 0 & 0 & 0 & 1 
\end{array}
\right].
\end{align}

\noindent We use two properties of $L_k$ when 
$k=2^t$ (see $L_8$ above for an illustration):
\begin{itemize}
  \setlength\itemsep{0em}
  \item[(1)] Each $L_k$ consists of three non-zero blocks: two identical
diagonal blocks which is $L_{k/2}$ and a block above the diagonal
which we call $\Gamma_{k/2}$ (which is almost like $L_{k/2}$
turned upside down).\footnote{Note here we use $\Gamma$ 
for something different from 
the previous section.}
  \item[(2)] The first row of $\Gamma_{k/2}$ contains only zero's.  The
$(i{+}2)^{\rm th}$ row of $\Gamma_{k/2}$ is the $\left(\frac{k}{2}-i\right)^{\rm
  th}$ row of $L_{k/2}$ (where $0 \leq i \leq k/2-2$).
\end{itemize}

We first explain why these two properties hold, as illustrated in Figure
\ref{fig4133}.  Take the Pascal's
triangle (mod 2) with $k$ rows, and rotate the entries 90 degrees
counter-clockwise.  This gives the (nontrivial) $(i,j)$ entries of
$L_k$ when $i \geq j$ and $i-j$ is even.
The stated properties for $L_k$ primarily come from the fact that
Pascal's triangle (mod 2) with $k$ rows consists of 4 triangles of
$k/2$ rows, the middle one only has zero entries, and the other three
are identical copies of Pascal's triangle (mod 2) with $k/2$ rows.
Also, the triangle is always left-right symmetric.  Proofs of these
are readily obtained from Lucas' Theorem%
\footnote{Consider the base-$p$ representation of integers $m$ and $n$, 
where $m \geq n \geq 0$, and $p$ is prime:
  $m = m_0 + m_1p + \ldots + m_k p^k, 
  n = n_0 + n_1p + \ldots + n_k p^k.$ 
Then,
  $\binom{m}{n} \equiv
  \binom{m_0}{n_0}\binom{m_1}{n_1}\ldots\binom{m_k}{n_k} \bmod p$
}
\ \cite{Fine1947} (a more accessible proof can be found online at \cite{agnesscott}).

\begin{figure}[h!]
  \begin{center}
  \includegraphics[width=0.6\columnwidth]{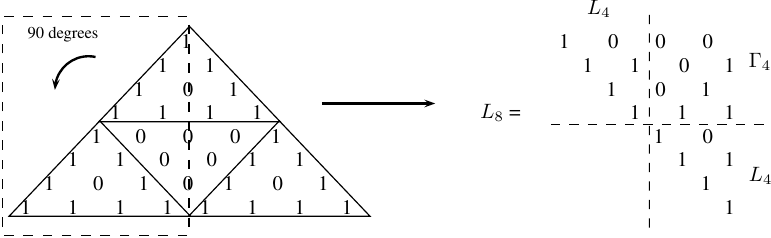}
\caption{\small An illustration of the Pascal's triangle structure of the
$L_8$ matrix. Taking the left half of an $8$-level Pascal's triangle and 
rotating
counter-clockwise by 90 degrees, we obtain the $L_8$ matrix. Note that the
block $\Gamma_{4}$ is the horizontal reflection of the lower diagonal
block $L_{4}$ with a downward shift, as described by property (2).
\label{fig4133}}
\end{center}
\end{figure}

If we multiply $L_k$ to a vector, 
\be
  \left[
	\begin{array}{c|c}
	  L_{k/2} & \Gamma_{k/2} \\ \hline
	  0 & L_{k/2}
  \end{array}\right]
  \left[\begin{array}{c}
	  v_1 \\
	  v_2
  \end{array}\right]
  = 
  \left[\begin{array}{c}
	  L_{k/2} \; v_1 + \Gamma_{k/2} \; v_2 \\
	  L_{k/2} \; v_2
  \end{array}\right].
\ee
Due to the relation between $\Gamma_{k/2}$ and $L_{k/2}$, the above mapping 
can be induced by the unitary $\mathscr{L}_k$ implemented by  
the circuit in Figure~\ref{fig4131}.  
Using standard recursion analysis, the circuit contains $O(k\log k)$
CNOT gates. 
\begin{figure}[h!]
  \begin{center}
  \includegraphics[width=0.4\columnwidth]{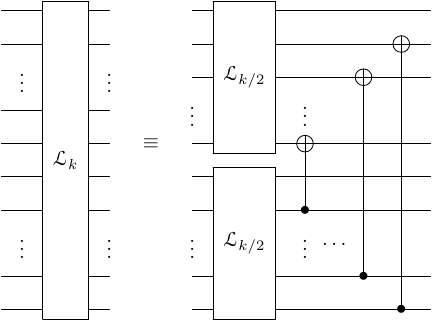}
\caption{\small An example of representation conversion circuit which demonstrates 
the recursive structure.\label{fig4131}}
\end{center}
\end{figure}

For general values of $k$, let $t = \lceil \hspace*{0.3ex} \log_2 k
\rceil$ and apply the above construction to obtain $\mathscr L_{2^t}$.
We restrict the circuit for $\mathscr L_{2^t}$ to a sub-circuit with the first $k$ registers and
the CNOT gates between them to obtain a circuit for $\mathscr{L}_k$ that still has size $O(k\log k)$.

A circuit for $\mathscr{L}_k^{-1}$ converting a vector from the self-dual
representation to the polynomial representation can be obtained by running
the circuit for $\mathscr{L}_k$ backwards.  The first qubit which corresponds to
the additional ``1'' in $\mathcal{S'}$ is always
$\left|0\right>$ and it remains untouched during the computation.
Therefore, the first qubit can be safely removed in the
circuit.  It is kept in the analysis for conceptual 
simplicity.  

Finally, we are ready to give the recipe for the fast multiplication
of two elements $a,r\in\mathrm{GF}(2^n)$ represented in the basis
$\mathcal{S'}$:
\begin{enumerate}
  \item Insert a zero at the beginning of the vector representations
    of $a$ and $r$ to get the vectors $a'$ and $r'$ with respect to
    the spanning set $\mathcal{S'}$.
  \item Convert $a'$ and $r'$ to new representations $\tilde{a}$ and
    $\tilde{r}$ with respect to the polynomial spanning set
    $\mathcal{T}$, using the circuit for $\mathscr{L}_{n+1}^{-1}$.  
  \item Multiply $\tilde{a}$ by $\tilde{r}$ using Sch{\"o}nhage's
    multiplication algorithm \cite{schonhage1977} 
    (denoted by $\tilde{\Pi}_r$ in figure \ref{fig4132}). The result is a vector
    with respect to the polynomial spanning set $\{1,
    \beta+\beta^{-1},(\beta+\beta^{-1})^2,\ldots,(\beta+\beta^{-1})^{2n}\}$.
  \item 
    Apply the unitary $\mathscr{L}_{2n+1}$ to the vector above 
    so it is represented in the spanning set $\{1, \beta+\beta^{-1},
    \beta^2+\beta^{-2}, \ldots, \beta^{2n}+\beta^{-2n}\}$.
    Then, discard the first element which is always 0.
    The result is the vector representation with respect
    to the spanning set $\{\beta+\beta^{-1}, \beta^2+\beta^{-2}, \ldots,
    \beta^{2n}+\beta^{-2n}\}$. Since $\beta$ is the $(2n+1)$-th root of
	unity in $\mathrm{GF}(2^{2n})$ (i.e., $\beta^{2n+1} = 1$), we have
    $\beta+\beta^{-1} = \beta^{2n}+\beta^{-2n}$, $\beta^2+\beta^{-2} =
    \beta^{2n-1}+\beta^{-2n+1}, \ldots$. Therefore with
    $n$ additional $\mathrm{GF}(2)$ CNOTs, the resulting vector can be
    reduced to the one with respect to the permuted self-dual normal basis
    $\mathcal{S}$.
\end{enumerate}
In Step 3, Sch{\"o}nhage's multiplication algorithm \cite{schonhage1977}
uses a radix-3 FFT algorithm to do fast convolution. Readers not familiar
with German may refer to \cite{von1999} for another description of
Sch{\"o}nhage's algorithm. This multiplication algorithm requires $O(n\log
n\log\log n)$ operations (additions and multiplications).  Additions  can
be implemented with CNOT gates.  Multiplications involved in this radix-3
FFT are the ones between an element of the polynomial ring 
$\mathrm{GF}(2)[x]/\left<x^{2m}+x^m+1\right>$ (for certain $m$) and $x$
(which is a $3m$-th root of unity in
$\mathrm{GF}(2)[x]/\left<x^{2m}+x^m+1\right>$). The result of this kind of
multiplications is a shift of coefficients and it can be implemented by
SWAP gates. Therefore, the whole multiplication method can be implemented
with $O(n\log n\log\log n)$ CNOT gates. As an example, Figure \ref{fig4132}
shows the implementation of $\Pi_r$ in $\mathrm{GF}(2^5)$.

\begin{figure}[h!]
  \begin{center}
	\includegraphics[width=0.6\columnwidth]{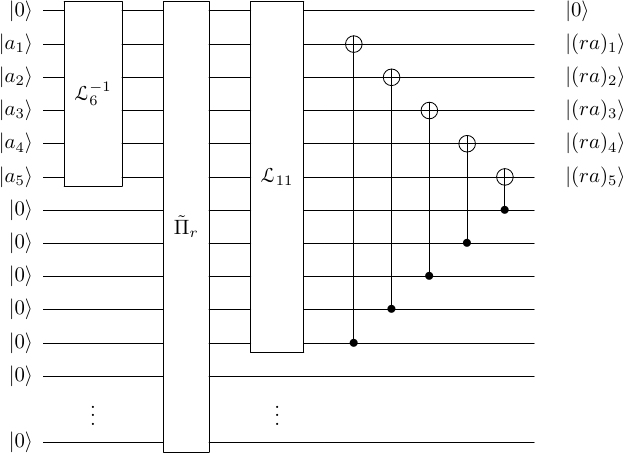}
\caption{\small The implementation of $\Pi_r$ for multiplication of $a$ by $r$ where
$a,r\in\mathrm{GF}(2^5)$. $\tilde{\Pi}_r$ is an implementation of Sch{\"o}nhage's
multiplication algorithm. The input and output bits are with respect to
a self-dual basis.\label{fig4132}}
\end{center}
\end{figure}

It is easy to show that the radix-3 FFT algorithm has logarithmic depth: if
the current step of this algorithm is working on a polynomial of degree
$k$, in the next recursion step, it will work in parallel on three
polynomials of degree $\lceil k/3 \rceil$. The total number of steps (i.e.,
the depth of the circuit) is therefore $O(\log n)$ for a polynomial of
degree $n$.  To multiply two polynomials of degree at most $n$, each
recursion step essentially consists of three components: computing the
radix-3 FFT, recursively doing $\lceil \sqrt{n} \, \rceil$ multiplications of
polynomials of degree at most $\lceil \sqrt{n} \, \rceil$ (in parallel), and
computing the inverse radix-3 FFT. Using a similar analysis, the depth of
the polynomial multiplication circuit is $O(\log(n) + \log(n^{1/2}) +
\log(n^{1/4}) + \ldots + 1) = O(\log n)$.  The logarithmic depth of the
basis conversion circuit can be shown by its recursive structure (e.g.,
Figure 4). Therefore, the depth of the circuit for $\Pi_r$ is $O(\log n)$.

The ancillary qubits can be reset to $\ket{0}$ using standard
techniques in reversible computing.  The result is a circuit for 
$\Pi_r$ for any non-zero $r \in \GF(2^n)$ with $O(n\log n\log\log n)$
CNOT gates .

\section{$\widetilde{O}(n)$ implementations based on polynomial basis for 
$\GF(2^n)$}\label{sec:mult-poly}

In this section we present alternative circuit constructions for
unitary $2$-designs in terms of \textit{polynomial} bases for
$\GF(2^n)$.
The advantage of using polynomial bases is that the $\SL_2(\GF(2^n))$
generator
$\bigl(\begin{smallmatrix}
	  r & 0\\ 0 & r^{-1}
\end{smallmatrix}\bigr)$ for $r \neq 0$ 
is straightforward to implement%
\footnote{
In this section, ``implement a mapping'' abridges ``implement the
unitary that induces a mapping according to Definitions~\ref{def:induce} or
\ref{def:induce2}'' and so on.} 
with $O(n \log n \log\log n)$ Clifford 
gates with depth $O(\log n)$, as described at the end of Section~\ref{sec:mult}.

For the generator, 
$\bigl(\begin{smallmatrix}
	  1 & 0\\ 1 & 1
\end{smallmatrix}\bigr)$ 
we provide two different $\widetilde{O}(n)$ circuit implementations 
in Subsections \ref{sec:nonclif} and \ref{sec:clif}.
However, we do not currently know how to implement the last 
generator  
$\bigl(\begin{smallmatrix}
	  0 & 1\\ 1 & 0
\end{smallmatrix}\bigr)$  
using only $\widetilde{O}(n)$ gates.
To circumvent this problem, we modify our ensemble for
the unitary $2$-design slightly.
Instead of implementing every element of $\SL_2(\GF(2^n))$, we
implement the elements that are lower triangular (i.e., $\LT_2(\GF(2^n))$), and we do this using $\widetilde{O}(n)$ gates. 
This follows directly from combining the implementations for 
$\bigl(\begin{smallmatrix}
	  r & 0\\ 0 & r^{-1}
\end{smallmatrix}\bigr)$ 
and
$\bigl(\begin{smallmatrix}
	  1 & 0\\ 1 & 1
\end{smallmatrix}\bigr)$ 
and by using Lemma \ref{lemma:gaussian_elim-triangular}.  
We can also implement all $M \in \UT_2(\GF(2^n))$ {\em with respect to
  the dual basis} (we denote this unitary $\widehat{U}_M$), 
because with respect to the dual basis, the operations that induce 
$\bigl(\begin{smallmatrix}
	  r & 0\\ 0 & r^{-1}
\end{smallmatrix}\bigr)$
and
$\bigl(\begin{smallmatrix}
	  1 & 1\\ 0 & 1
\end{smallmatrix}\bigr)$
are $H^{\otimes n} \Pi_r H^{\otimes n}$ and $H^{\otimes n} \Gamma_V
H^{\otimes n}$ (respectively).  
In Subsection~\ref{sec:bypass} we show how to combine the
implementations of $\LT_2(\GF(2^n))$ in the primal basis and
$\UT_2(\GF(2^n))$ in the dual basis to achieve Pauli mixing.
This results in an exact unitary $2$-design with the desired 
complexity.

\subsection[Implementation of 
ADD with $O(n \log n \log\log n)$ non-Clifford gates]{Implementation of 
$\bigl(\begin{smallmatrix}
	  1 & 0\\ 1 & 1
\end{smallmatrix}\bigr)$ with $O(n \log n \log\log n)$ non-Clifford gates}
\label{sec:nonclif}

Here we provide an implementation of 
$\bigl(\begin{smallmatrix}
	  1 & 0\\ 1 & 1
\end{smallmatrix}\bigr)$
using
$O(n \log n \log\log n)$ gates that can be organized so as to have depth 
$O(\log n)$.
This construction uses non-Clifford gates but they compose to a 
Clifford unitary.  
(The next subsection contains a slightly less efficient construction using 
only Clifford gates.)

The operation that we need to implement is $\Gamma_W$, defined in
Eqs.~\eqref{eq:Gamma_V} and \eqref{eq:Gamma_V-alternate} (with $V$ set
to $W$).
Recall that $W$ is the primal-to-dual basis conversion matrix of
Eq.~\eqref{basis-convert}.  
Since we are setting the primal basis to a polynomial basis, $W$ is a
\textit{Hankel matrix}: for all $j, k, j', k'$, if $j+k = j' + k'$
then $W_{jk} = W_{j'k'}$.  We make use of this property in this and 
the next subsection.
From Eq.~\eqref{eq:Gamma_V}, 
\begin{align}\label{eq:secondU_W}
\Gamma_W \ket{c} 
&=
\mbox{\Large $i$}^{\,\sum_{j=1}^{n} \sum_{k=1}^{n} W_{jk}c_jc_k}\ket{c}.
\end{align}
Note that it suffices to compute the exponent of $i$ using mod $4$
arithmetic, and the exponent has the form
\begin{align}
\label{iexponent}
\mbox{\Large\bf [}
\, c_{1}\, \cdots\, c_{n}\,
\mbox{\Large\bf ]}\,
W
\begin{bmatrix}
c_{1} \\
\vdots \\
c_{n}
\end{bmatrix}\,.
\end{align}

This problem is related to the problem of computing convolutions.
Recall that the \textit{convolution} of two $d$-dimensional vectors $u$ and $v$ is
defined as the $(2d-1)$-dimensional vector $w$ such that 
\begin{align}
w_0 + &w_1 T + w_2^2 T^2 + \cdots + w_{2d-2}T^{2d-2} \\
&= \left(u_0 + u_1 T
+ v_2^2 T^2 + \cdots + u_{d-1}T^{d-1}\right)\left(v_0 + v_1 T + v_2^2 T^2 + \cdots
+ v_{d-1}T^{d-1}\right)
\end{align}
as polynomials over $T$.
The product of a Hankel matrix with a vector reduces to convolution, 
as shown in the next proposition.

\begin{prop}\label{prop:convolution}
The product of an $n \times n$ Hankel matrix with an $n$-dimensional vector reduces to the problem of computing the convolution of two $(2n-1)$-dimensional vectors.
\end{prop}

\begin{proof}
This can be seen by comparing 
\begin{align}\label{eq:convolution}
\begin{bmatrix}
x_1    & x_2     & \cdots & x_n \\
x_2    & x_3     & \cdots & x_{n+1} \\
\vdots & \vdots  & \ddots & \vdots \\
x_n    & x_{n+1} & \cdots & x_{2n-1}
\end{bmatrix}
\begin{bmatrix}
y_{1} \\
y_{2} \\
\vdots \\
y_{n}
\end{bmatrix}
\end{align}
with the middle components of the convolution of $[x_1, \dots, x_{2n-1}]$ and 
$[0, \dots, 0, y_n, \dots, y_1]$.
The convolution is a $(4n-3)$-dimensional vector that is the vector 
in Eq.~\eqref{eq:convolution} padded with $2n-2$ components on the left and $n-1$ components on the right. 
\end{proof}

\noindent
Returning to the computation of Eq.~(\ref{iexponent}), we can compute 
$e_1, \dots, e_{n} \in \mathbb{Z}_4$, given by  
\begin{align}
\begin{bmatrix}
e_{1} \\
\vdots \\
e_{n}
\end{bmatrix}
=
W
\begin{bmatrix}
c_{1} \\
\vdots \\
c_{n}
\end{bmatrix},
\end{align}
with a fast algorithm for polynomial multiplication%
\footnote{We conjecture that it is possible to slightly reduce the gate count for this construction from $O(n \log n \log \log n)$ to $(n \log n)2^{O(\log^* n)}$ by employing the improved algorithms for integer multiplication initiated by
F\"{u}rer~\cite{Furer2009}.} 
over the ring
$\mathbb{Z}_4$ using only $O(n \log n \log\log n)$ gates 
(see, for example, Theorem~8.23 in~\cite{von1999}).
Then Eq.~(\ref{iexponent}) for the exponent for $i$ in $\Gamma_W$ can
be obtained from the $2n$ ancillary qubits containing $e_1, \dots,
e_{n}$ (each $e_j$ is a two-bit string) and the $n$ qubits containing
$c_{1}, \dots, c_{n}$ as follows.  For each $j \in \{1,\dots,n\}$,
apply a controlled-$Z$ gate between the high order bit of $e_j$ and
$c_{j}$ and apply a controlled-$S$ gate between the low order bit of
$e_j$ and~$c_{j}$.

This construction explicitly uses the non-Clifford controlled-$S$
gates, since the underlying ring is $\mathbb{Z}_4$ and addition mod 4
requires non-Clifford gates.
The construction uses polynomial multiplications, so it follows
from the circuit depth analysis in Section~\ref{sec:mult} that the
circuit depth of this construction is $O(\log n)$.  In the next
subsection, we describe a different procedure for implementing
$\bigl(\begin{smallmatrix} 1 & 0\\ 1 & 1 \end{smallmatrix}\bigr)$ that
is slightly less efficient, but uses only Clifford gates.

\subsection[Implementation of 
ADD with $O(n \log^2 n \log\log n)$ Clifford gates]{Implementation of 
$\bigl(\begin{smallmatrix}
	  1 & 0\\ 1 & 1
\end{smallmatrix}\bigr)$ with $O(n \log^2 n \log\log n)$ Clifford gates}
\label{sec:clif}

Here we provide an implementation of 
$\bigl(\begin{smallmatrix}
	  1 & 0\\ 1 & 1
\end{smallmatrix}\bigr)$
using $O(n \log^2 n \log\log n)$ Clifford gates that can be organized
so as to have depth $O(\log^2 n)$.
%
%
In the previous subsection, the computation is reduced to a
convolution in mod 4 arithmetic, and we needed non-Clifford gates to
compute this efficiently.  Here, we use a recursive procedure that is
based on convolutions in mod 2 arithmetic, which can be performed
efficiently with Clifford gates.  We assume all notation from 
the previous subsection.  

To simplify our presentation, we assume that $n$ is a power of 2 (though our approach can be generalized to arbitrary $n$ by dividing unevenly in the recursive step, as $n = \lfloor \frac n2 \rfloor + \lceil \frac n2 \rceil$).
We divide $W$ into four $\frac n2 \times \frac n2$ blocks as
\begin{align}
W = \begin{bmatrix}
W^{(11)} & W^{(12)} \\
W^{(21)} & W^{(22)}
\end{bmatrix}
\end{align}
where $W^{(11)}, W^{(12)}, W^{(21)}, W^{(22)}$ are $\frac n2 \times \frac n2$ Hankel matrices and $W^{(12)} = W^{(21)}$.
Define
\begin{align}
A = \begin{bmatrix}
0 & W^{(12)} \\
W^{(21)} & 0
\end{bmatrix}, \ \ \ 
B = \begin{bmatrix}
W^{(11)} & 0 \\
0 & 0
\end{bmatrix}, \ \ \ 
C = \begin{bmatrix}
0 & 0 \\
0 & W^{(22)}
\end{bmatrix}.
\end{align}
Clearly, 
\begin{align}
\begin{bmatrix}
I & 0 \\ W & I
\end{bmatrix}
=
\begin{bmatrix}
I & 0 \\ A & I
\end{bmatrix}
\begin{bmatrix}
I & 0 \\ B & I
\end{bmatrix}
\begin{bmatrix}
I & 0 \\ C & I
\end{bmatrix}\end{align}
so we can implement
$\Gamma_A$, $\Gamma_B$, and $\Gamma_C$ 
separately, and compose them to obtain $\Gamma_W$.

We first show how to implement $\Gamma_A$
using $O(n \log n \log\log n)$ gates.
From Eq.~\eqref{eq:Gamma_V-alternate}, 
\begin{align}\label{eq:rec-phases}
\Gamma_A\ket{c} = 
(-1)^{\sum_{j=1}^{n/2} \sum_{k=n/2+1}^{n} W_{jk} c_j c_k}\ket{c} \,.
\end{align}
The expression for the exponent of ${-}1$ above can be computed in 
mod 2 arithmetic, and has the form
\begin{align}
\mbox{\Large\bf [}\,
c_{1} \,
\cdots\, c_{\frac n2}\,
\mbox{\Large\bf ]}\,
W^{(12)}
\begin{bmatrix}
c_{\frac n2+1}\\
\vdots\\
c_{n}
\end{bmatrix}\,.
\end{align}
Once again, by Proposition~\ref{prop:convolution}, the above product 
of a Hankel matrix with a vector reduces to convolution, and hence 
polynomial multiplication over the field $\GF(2)$.  
We can compute the bits $e_{\frac n2+1}, \dots, e_n$, defined as 
\begin{align}
\begin{bmatrix}
e_{\frac n2+1} \\
\vdots \\
e_{n}
\end{bmatrix}
=
W^{(12)}
\begin{bmatrix}
c_{\frac n2+1} \\
\vdots \\
c_{n}
\end{bmatrix}
\end{align}
in $\frac n2$ ancillary registers using only $O(n \log n \log\log n)$ gates.
Moreover, since the convolution is with respect to entries of $W$---which are constants in our setting---all the gates can be Clifford gates (in fact, CNOT gates).
Then we can apply $O(n)$ controlled-$Z$ gates between the bits $e_{\frac n2+1}, \dots, e_{n}$ and $c_{1}, \dots, c_{\frac n2}$ (respectively) to apply the phase that correctly implements $\Gamma_A$.

What remains is to compute $\Gamma_B$ and $\Gamma_C$.
Each of these is equivalent to computing an instance of the original problem
of size $n/2$.
In the bottom of the recurrence (when $W$ is a $1 \times 1$ matrix), a single $S$ (phase) gate computes $\Gamma_W$.
The gate cost $G(n)$ of the recursive procedure satisfies the recurrence 
\begin{align}
G(n) = 2\,G(n/2) + O(n \log n \log\log n),
\end{align}
whose solution satisfies
\begin{align}
G(n) \in O(n \log^2 n \log\log n).
\end{align}
This recursive construction needs polynomial multiplication in each
recursion step. According to the circuit depth analysis for polynomial
multiplication in Section~\ref{sec:mult}, the circuit depth is $O(\log n + \log \frac{n}{2} + \ldots + 1) = O(\log^2n)$.

\subsection{Pauli mixing from $\LT_2(\GF(2^n))$ and 
$\UT_2(\GF(2^n))$ in different bases}\label{sec:bypass}

Here, we show how to achieve Pauli mixing by implementing $U_M$ for 
$M \in \LT_2(\GF(2^n))$ and $\widehat{U}_M$ for $M \in \UT_2(\GF(2^n))$.
We will explain our approach in two parts. In the first part, we
explain the actual generation and construction of the ensemble of
unitaries---which is simple, but the resulting ensemble no longer
corresponds to $\SL_2(\GF(2^n))$, so it is not clear that the ensemble is a 
unitary 2-design.
In the second part, we prove that the new ensemble is Pauli mixing,
so it is indeed a unitary 2-design.

The construction is based on the following decomposition of elements of 
$\SL_2(\GF(2^n))$, along the lines of Eq.~\eqref{eq:decomposition}: 
\begin{align}\label{eq:decomposition2}
\begin{pmatrix} \alpha & \gamma \\ \beta & \delta
\end{pmatrix} = 
\begin{cases}
\begin{pmatrix} 1 & 0\\
 \frac\beta\alpha&
1\end{pmatrix}
\begin{pmatrix}
\alpha & \gamma\\
0 & \alpha^{-1}\end{pmatrix}
& \mbox{if $\alpha \neq 0$} \vspace*{2mm}\\
\begin{pmatrix}
\gamma & 0\\ 
\delta & \gamma^{-1}
\end{pmatrix}
\begin{pmatrix}
0 & 1 \\
1 & 0
\end{pmatrix} 
& \mbox{if $\alpha = 0$.} 
\end{cases}
\end{align}
Note that all matrices in this decomposition are lower triangular, upper triangular, or 
$\bigl(\begin{smallmatrix}
0 & 1 \\ 1 & 0
\end{smallmatrix}\bigr)$.
Lower triangular matrices can be implemented in the primal basis; upper triangular matrices can be implemented in the dual basis; and 
$\bigl(\begin{smallmatrix}
0 & 1 \\ 1 & 0
\end{smallmatrix}\bigr)$ 
can be implemented in any self-dual basis (by $H^{\otimes n}$).

The procedure to generate an element of the ensemble is as follows.

\medskip

\begin{quote}
\noindent\textbf{Generation procedure:}
\begin{description}
\item[1.\ \,]
Sample $\bigl(\begin{smallmatrix}
\alpha & \gamma \\ \beta & \delta
\end{smallmatrix}\bigr) \in \SL_2(\GF(2^n))$
according to the uniform distribution.
\item[2.1] If $\alpha \neq 0$ then \newline
\hspace*{4mm} set $M_1$ to 
$\bigl(\begin{smallmatrix}
\alpha & \gamma \\  0 & \,\alpha^{-1}
\end{smallmatrix}\bigr)$ \vspace*{0.5mm} \newline
\hspace*{4mm} set $M_2$ to 
$\bigl(\begin{smallmatrix}
1 & 0 \\ \beta/\alpha\, & 1
\end{smallmatrix}\bigr)$ \newline
\hspace*{4mm} construct the Clifford group element 
$U_{M_2} \circ \widehat{U}_{M_1}$ (composition of two circuits).
\item[2.2] Else if  $\alpha = 0$ then \newline
\hspace*{4mm} set $M$ to $\bigl(\begin{smallmatrix}
\gamma & 0 \\ \delta & \gamma^{-1}
\end{smallmatrix}\bigr)$ \newline
\hspace*{4mm} construct the Clifford group element 
$U_{M} \circ H^{\otimes n}$ (composition of two circuits).
\end{description}
\end{quote}
Note that the composition in step 2.1 is along the lines of the first case of Eq.~\eqref{eq:decomposition2} and the composition in step 2.2 is along the lines of the second case of Eq.~\eqref{eq:decomposition2}.
In each case, a Clifford group element with gate complexity 
$O(n \log n \log\log n)$ (or Clifford-gate complexity $O(n \log^2 n \log\log n)$) results; however, the subset of all Cliffords that can arise by this procedure does not have the structure of $\SL_2(\GF(2^n))$ because of the disparate coordinate systems being used for the components.
This concludes the description of the generation and construction of elements of the ensemble.

We now explain why the ensemble resulting from the above procedure is a unitary 2-design in spite of the mismatched bases used to convert the matrices arising from Eq.~\eqref{eq:decomposition2} into Clifford unitaries.
First, we consider the mixing property over the Paulis that results from $U_M$ for a random $M \in \LT_2(\GF(2^n))$, and similarly for $\UT_2(\GF(2^n))$.
Partition the non-zero elements of $\GF(2^n) \times \GF(2^n)$ into these two (disjoint) subsets:
\begin{align}
R_1 &= \bigl\{
\bigl(\begin{smallmatrix}
a \\ b 
\end{smallmatrix}\bigr)
\in \GF(2^n)\times \GF(2^n) : \mbox{$a = 0$ and $b \neq 0$} 
\bigr\} \\[1mm]
R_2 &= \bigl\{
\bigl(\begin{smallmatrix}
a \\ b 
\end{smallmatrix}\bigr)
\in \GF(2^n)\times \GF(2^n) : \mbox{$a \neq 0$} 
\bigr\}.
\end{align}
It is straightforward to verify that a random element $M \in \LT_2(\GF(2^n))$ uniformly mixes within $R_1$ and it uniformly mixes within $R_2$ in the following sense.
\begin{lemma}\label{lemma-LT-mixing}
Let $M\in \LT_2(\GF(2^n))$ be chosen uniformly at random. 
Then, for any 
$\bigl(\begin{smallmatrix}
	  a \\ b
\end{smallmatrix}\bigr)
\in R_1$, 
the distribution
$M\bigl(\begin{smallmatrix}
	  a \\ b
\end{smallmatrix}\bigr)$
is uniform over $R_1$
and, for any 
$\bigl(\begin{smallmatrix}
	  a \\ b
\end{smallmatrix}\bigr)
\in R_2$, 
the distribution
$M\bigl(\begin{smallmatrix}
	  a \\ b
\end{smallmatrix}\bigr)$
is uniform over $R_2$. 
\end{lemma}
A similar result holds for $\UT_2(\GF(2^n))$ with $a$ and $b$ switched in the definitions of $R_1$ and $R_2$ (we omit the simple proof of this).

To illustrate the consequences of Lemma~\ref{lemma-LT-mixing} on the
Paulis, we can organize the $n$-qubit Paulis into rows and columns
where $X^{\lceil a \rceil} Z^{\lfloor b \rfloor}$ is in column $a$ and
row $b$.  We choose the first row and column to be labeled by 
$a=0$ and $b=0$ and call them the {\em zero row} and {\em zero column}.
The relative ordering of the remaining rows and columns does not affect our discussion; they are collectively called the {\em nonzero rows} and the {\em nonzero columns}.  
Figure~\ref{fig:pauli-mix} shows such a layout for the $n=2$ case where the identity Pauli is excluded.
\begin{figure}[h!]
{\small
\begin{center}
\vspace*{1mm}
\begin{tabular}{cccc}
     & \!$IX$\! & \!$XI$\! & \!$XX$ \vspace*{1mm}\\
$IZ$\! & \!$IY$\! & \!$XZ$\! & \!$XY$ \vspace*{1mm}\\
$ZI$\! & \!$ZX$\! & \!$YI$\! & \!$YX$ \vspace*{1mm}\\
$ZZ$\! & \!$ZY$\! & \!$YZ$\! & \!\!\!$YY$
\end{tabular}
\vspace*{-5mm}
\end{center}
}
\caption{\small A natural arrangement of all the non-trivial 2-qubit Paulis into rows and columns. Pauli mixing requires a uniform distribution on the $15$ items.}\label{fig:pauli-mix}
\end{figure}

Based on Lemma~\ref{lemma-LT-mixing}, conjugating by $U_M$ for a uniformly
distributed $M \in \LT_2(\GF(2^n))$ causes the zero column to mix
uniformly and also the complement of the zero column (consisting of all 
the nonzero columns) to mix uniformly.
We call this effect \textit{lower-triangular Pauli mixing}. Schematically, this is illustrated in Figure~\ref{fig:lower-triangular-mixing}.
\begin{figure}[h!]
\centering
\includegraphics[width=0.4\textwidth]{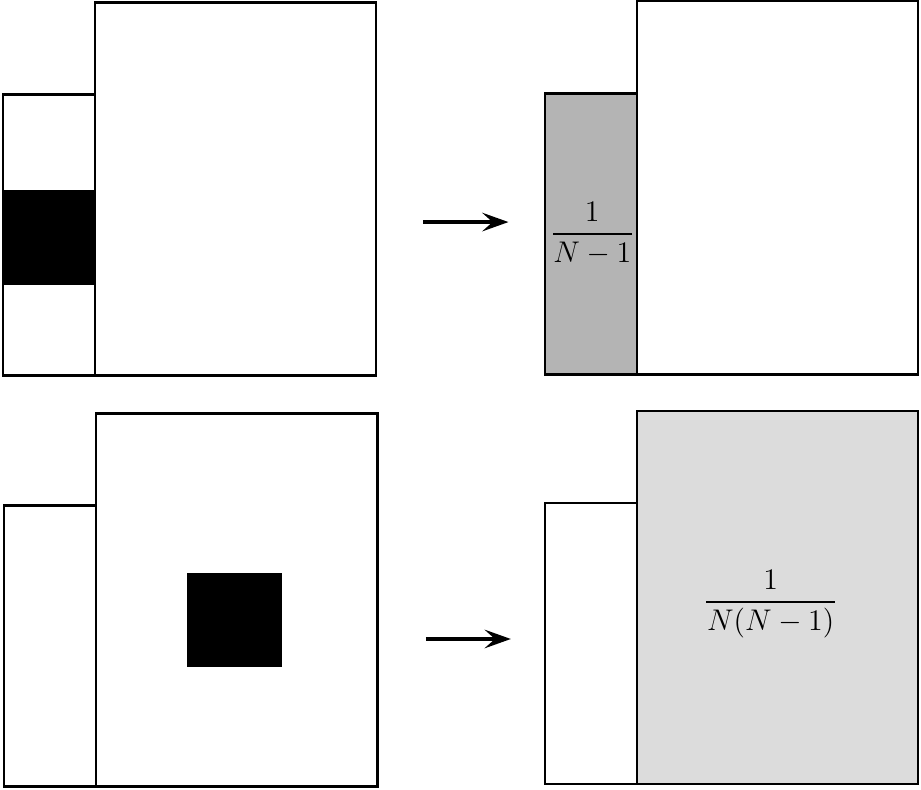}
\caption{\small Illustration of lower-triangular Pauli mixing. Top: mixing effect within the zero column. Bottom: mixing effect within the complement of the 
zero column ($N=2^n$).}\label{fig:lower-triangular-mixing}
\end{figure}
We can similarly define \textit{upper-triangular Pauli mixing}, corresponding to a transposed version of Figure~\ref{fig:lower-triangular-mixing}.
Sampling 
$M \in \UT_2(\GF(2^n))$ and then constructing the Clifford unitary 
$\widehat{U}_M$ achieves upper-triangular mixing.

We define one additional form of mixing, that we call \textit{column Pauli mixing}, illustrated in Figure~\ref{fig:column-mixing}, where Paulis in the zero column do not change and any Pauli in a nonzero column mixes within its column.
Such mixing is accomplished by choosing
$M = \bigl(\begin{smallmatrix}
1 & 0 \\ \beta & 1
\end{smallmatrix}\bigr)$
for a uniformly random $\beta \in \GF(2^n)$, and then constructing the Clifford unitary $U_M$.

\begin{figure}[h!]
\centering
\includegraphics[width=0.4\textwidth]{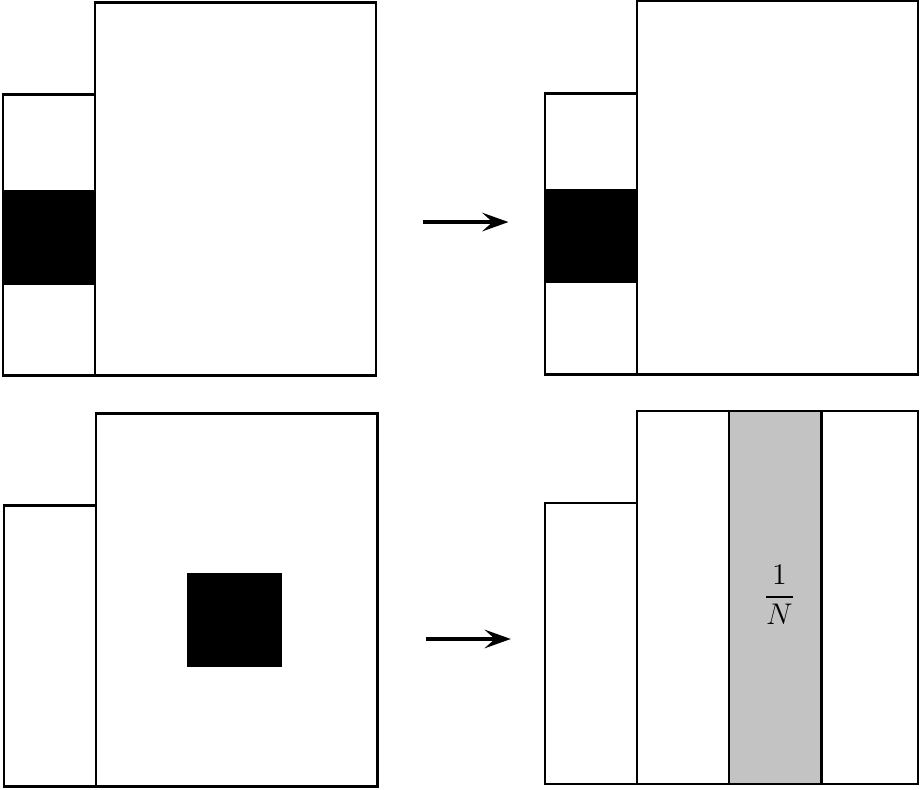}
\caption{\small Illustration of column mixing. Top: elements in the zero column stay put. Bottom: elements in any nonzero column uniformly mix within the column ($N=2^n$).}\label{fig:column-mixing}
\end{figure}

From Eq.~\eqref{eq:decomposition2}, we can deduce that our procedure is applying a probabilistic mixture of the two procedures below.
With probability $\frac{2^n}{2^n+1}$ it applies Procedure A; with probability 
$\frac{1}{2^n+1}$ it applies Procedure B ($\frac{1}{2^n+1}$ is the probability that $\alpha = 0$ for a random 
$\bigl(\begin{smallmatrix}
\alpha & \gamma \\ \beta & \delta
\end{smallmatrix}\bigr) \in \SL_2(\GF(2^n))$).


\begin{quote}
\noindent\textbf{Procedure A}:
\begin{enumerate}
\item Apply an upper-triangular Pauli mixing operation.
\item Apply a column Pauli mixing operation (independently from the first step).
\end{enumerate}
\end{quote}


\begin{quote}
\noindent\textbf{Procedure B}: 
\begin{enumerate}
\item Apply $H^{\otimes n}$ (thereby transposing the layout of the Paulis).
\item Apply a lower-triangular mixing operation.
\end{enumerate}
\end{quote}

We now prove that the above mixture of Procedures A and B results in
Pauli mixing.

\begin{lemma}
The stochastic process of applying either Procedure A or Procedure B, with probabilities $\frac{2^n}{2^n+1}$ and $\frac{1}{2^n+1}$ (respectively) is Pauli mixing.
\end{lemma}

\begin{proof}
For convenience, let $N = 2^n$.
First, consider an initial Pauli in the zero row (i.e., $b=0$ and it is of the form $X^{\lceil a \rceil}$ for some $a\neq 0$).
Then, as illustrated in Figure~\ref{fig:first-row-mixing},
\begin{figure}[h!]
\centering
\includegraphics[width=0.6\textwidth]{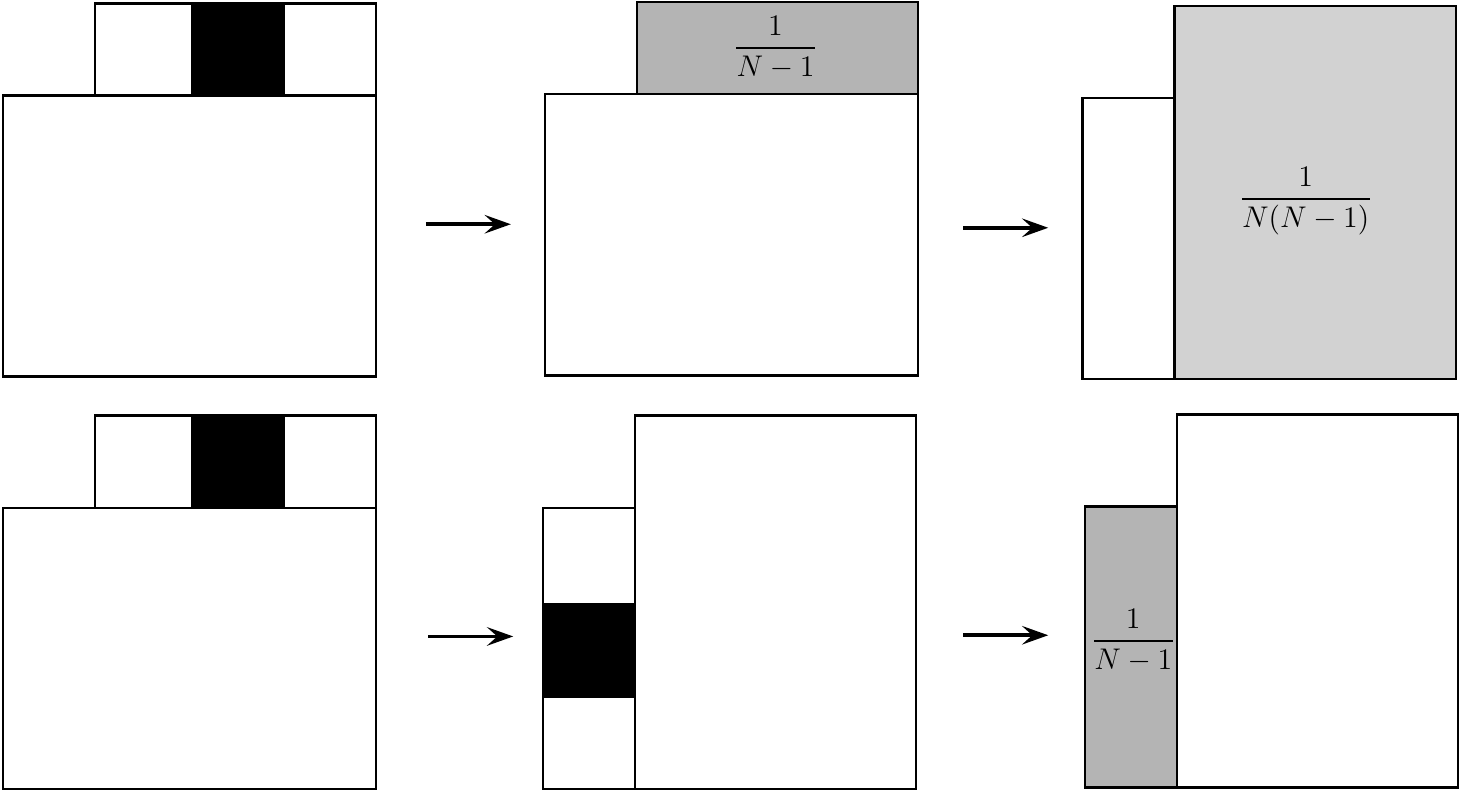}
\caption{\small Illustration of mixing procedure starting in the zero row  ($N=2^n$). Top: Procedure~A. Bottom: Procedure~B.}\label{fig:first-row-mixing}
\end{figure}
if Procedure A is applied, the result is a uniform distribution over all nonzero columns, where the probability of each Pauli is 
$\frac{1}{N(N-1)}$.
On the other hand, if Procedure B is applied, the result is a uniform distribution on the zero column, where the probability of each Pauli is 
$\frac{1}{N-1}$.
Consider the mixture of these distributions (Procedure A with probability 
$\frac{N}{N+1}$ and Procedure B with probability 
$\frac{1}{N+1}$).
Since 
$\frac{N}{N+1}\frac{1}{N(N-1)} = \frac{1}{N^2-1}$ and
$\frac{1}{N+1}\frac{1}{N-1} = \frac{1}{N^2-1}$, 
the result is the uniform distribution.

Next, consider the case of an initial Pauli that is not in the zero row (i.e., $X^{\lceil a \rceil} Z^{\lfloor b \rfloor}$ with $b \neq 0$).
Then, as illustrated in Figure~\ref{fig:general-mixing},
\begin{figure}[h!]
\centering
\includegraphics[width=0.6\textwidth]{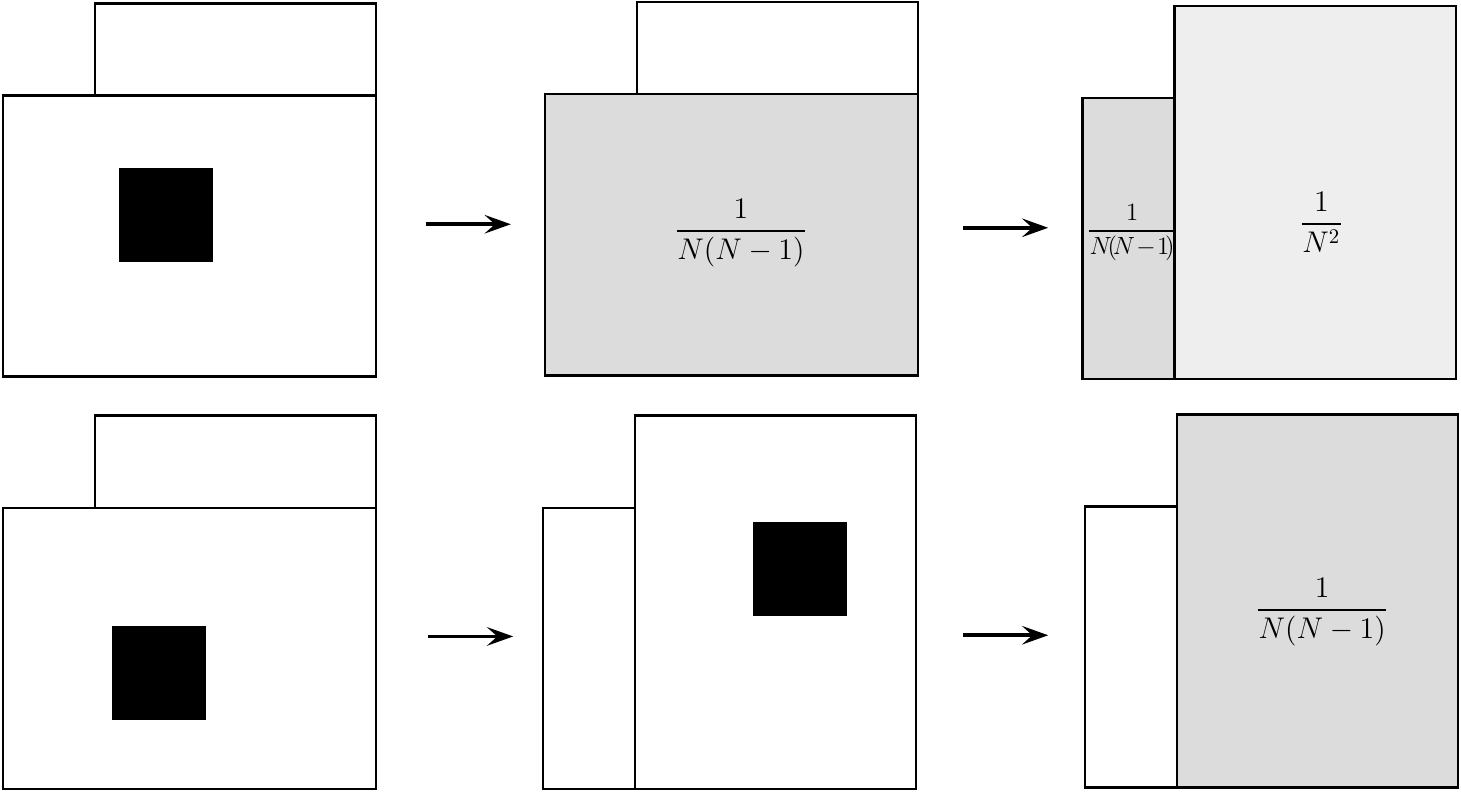}
\caption{\small Illustration of mixing procedure starting in a nonzero row ($N=2^n$). Top: Procedure~A. Bottom: Procedure~B.}\label{fig:general-mixing}
\end{figure}
if Procedure A is applied, the result is a two-level distribution: the probability of each Pauli in the zero column is $\frac{1}{N(N-1)}$; the probability of each Pauli in any nonzero column is 
$\frac{1}{N^2}$.
On the other hand, if Procedure B is applied, the result is a uniform
distribution over the nonzero columns, where the probability of each Pauli is $\frac{1}{N(N-1)}$.
Consider the mixture of these distributions (Procedure A with
probability
$\frac{N}{N+1}$ and Procedure B with probability 
$\frac{1}{N+1}$).
Since 
$\frac{N}{N+1}\frac{1}{N(N-1)} = \frac{1}{N^2-1}$ and  
$\frac{N}{N+1}\frac{1}{N^2}+\frac{1}{N+1}\frac{1}{N(N-1)} = \frac{1}{N^2-1}$, 
the result is the uniform distribution.
\end{proof}

\section{Acknowledgments}
We thank Olivia Di Matteo for pointing out reference \cite{GAE2007}
and Aram Harrow for pointing out reference \cite{Low10}.  We thank
Joachim von zur Gathen, Mark Giesbrecht, and Arne Storjohann for
discussions about finite field algorithms, and Daniel Gottesman and
Aram Harrow for other discussions.  This research was supported in
part by Canada's NSERC and CRC, a David R. Cheriton Graduate
Scholarship, and the U.S. ARO.

\bibliographystyle{abbrv}


\appendix

\section{Proof of Lemma \ref{lem:equivdef} and Corollary \ref{cor:dagger}}
\label{appendix:definitions}

\noindent {\bf Lemma~\ref{lem:equivdef}.}
{\it Let $\cE$ be any ensemble of unitaries in $\ug_N$.  
Then, the following are equivalent:
\begin{itemize}
  \setlength\itemsep{0em}
  \item[(1)] $\cE$ is degree-$2$ expectation preserving. 
  \item[(2)] $\cE$ is $2$-query indistinguishable. 
  \item[(3)] $\cE$ implements the full bilateral twirl. 
  \item[(4)] $\cE$ implements the full channel twirl.
\end{itemize}
}

\noindent {\bf Corollary \ref{cor:dagger}.}
{\it For $\cE = \set{p_i,U_i}_{i=1}^k$, let $\cE^\dagger {:}{=}
\set{p_i,U_i^\dagger}_{i=1}^k$. 
\begin{itemize}
  \setlength\itemsep{0em}
  \item[(a)] $\cE$ implements the full bilateral twirl if and only if $\cE^\dagger$ does. 
  \item[(b)] $\cE$ implements the full channel twirl if and only if $\cE^\dagger$ does.
\end{itemize}
}

\begin{proof}
We will show, in order, 
(1) $\Rightarrow$ (2) $\Rightarrow$ (3) $\Rightarrow$ (1), Corollary 
\ref{cor:dagger}(a), then, 
(2) $\Rightarrow$ (4) $\Rightarrow$ (3), and finally Corollary 
\ref{cor:dagger}(b).  

\vspace*{1.5ex}

\noindent (1) $\Rightarrow$ (2): 
Consider any distinguishing circuit $\mathcal C$ making up to two queries
of $U$ or $U^\dagger$.  Note that the output state $\eta_2(\mathcal C,U)$
is a product of matrices with at most two factors of $U$ and two
factors of $U^\dagger$.  Thus, each entry of $\eta_2(\mathcal C,U)$ is a
polynomial of degree at most $2$ in the matrix elements of $U$ and at
most $2$ in the complex conjugates of those matrix elements.
By hypothesis, $\cE$ is degree-$2$ expectation preserving, 
thus the following holds entrywise:
\be \sum_{i=1}^k \; p_i \, \eta_2(\mathcal C,U_i) = 
\int d\mu(U) \, \eta_2(\mathcal C,U) \,\ee
and $\cE$ is $2$-query indistinguishable. 

\vspace*{1.5ex}
\noindent (2) $\Rightarrow$ (3): This follows from the definition that
the bilateral twirl circuit is a special case of a 2-query
distinguishing circuit $\mathcal C$.

\vspace*{1.5ex}
\noindent (3) $\Rightarrow$ (1): Let $\{|j\>\}_{j=1}^N$ be a basis for
$\mathbb{C}^{N}$.  Suppose $\cE$ implements the full bilateral twirl,
so, $\forall \rho$,
\be
\sum_{i=1}^k \, p_i \, U_i \otimes U_i \, 
\rho \, U_i^\dagger \otimes U_i^\dagger 
= \int \, d \mu(U) \, U \otimes U \, \rho \, 
U^\dagger \otimes U^\dagger \,.
\label{eq:btwirlrepeat}
\ee
Since the density matrices span the complex Hilbert space of all
possible square matrices of the same dimension, the above relation holds
if we replace $\rho$ by $|a_1\>\<a_3| \otimes |a_2\>\<a_4|$, for all 
$a_1, a_2, a_3, a_4 \in \{1,\cdots,N\}$.  
Furthermore, we can left- and right-multiply the above equation by 
$\<a_5|\otimes\<a_6|$ and $|a_7\>\otimes|a_8\>$.  This gives 
\be
\sum_{i=1}^k \, p_i \<a_5|U_i|a_1\> \<a_6|U_i|a_2\> 
\<a_3|U_i^\dagger|a_7\> \<a_4 |U_i^\dagger|a_8\>  
= \int \, d \mu(U) \<a_5|U|a_1\>  \<a_6|U|a_2\>
\<a_3|U^\dagger|a_7\> \<a_4 |U^\dagger|a_8\> \,.
\nonumber
\ee
Repeating the above for all possible $a_1,\cdots,a_8$ and applying
linearity implies Eq.~(\ref{eq:u2design}) and that $\cE$ is degree-$2$
expectation preserving.


\vspace*{1.5ex}

\noindent Corollary \ref{cor:dagger}(a): From
Definition~\ref{def:2queryu}, $\cE$ is $2$-query indistinguishable iff
$\cE^\dagger$ is.  Thus, by the equivalence between (2) and (3), 
$\cE$ implements the full bilateral twirl if and only if $\cE^\dagger$ does.  


\vspace*{1.5ex}

\noindent (2) $\Rightarrow$ (4): This follows from the definition that
the channel twirl circuit is a special case of a $2$-query
distinguishing circuit $\mathcal C$.

\vspace*{1.5ex}

\noindent (4) $\Rightarrow$ (3): We provide a proof for the most
general unitary 2-design here.  Readers interested in the special 
(but common) case when the ensemble $\cE$ consists only 
of Clifford unitaries and $N=2^n$ can consult Appendix
\ref{appendix:shortproof} for a short proof.  

We begin with some relevant concepts in quantum information.  
Let $|1\>,\cdots,|N\>$ denote an orthonormal basis for $\mathbb{C}^N$,
${\cal B}(\mathbb{C}^N)$ denote the set of all bounded $N \times N$
matrices, and $\Phi = \sum_{l,j=1}^N |l\>\<j| \otimes |l\>\<j|$.
Let ${\cal I}$ denote the identity map on ${\cal B}(\mathbb{C}^N)$.  
For any linear map $\Theta: {\cal B}(\mathbb{C}^N) \rightarrow 
{\cal B}(\mathbb{C}^N)$, denote the \emph{Choi-matrix} of
$\Theta$ by $J(\Theta) = (\Theta \otimes {\cal I}) (\Phi) 
= \sum_{l,j=1}^N \Theta(|l\>\<j|) \otimes |l\>\<j|$ \cite{Choi75}.
$\Theta$ is completely positive if and only if 
$J(\Theta)$ is positive semidefinite \cite{Choi75} 
(see also \cite{Leung03,Watrous-notes}).   
A quantum channel is a linear, trace-preserving, and completely 
positive map.  

Suppose for every quantum channel $\Lambda$, $\mathbb{E}_\cE(\Lambda) =
\mathbb{E}_\mu(\Lambda)$.  Then, $J(\mathbb{E}_\cE(\Lambda)) = 
J(\mathbb{E}_\mu(\Lambda))$.  Rephrasing this equality  
using Eqs.\ (\ref{def:e-twirl-channel}) and (\ref{def:haar-twirl-channel}), 
we have 
\be
\sum_{i=1}^k \, p_i \, 
(U_i^\dagger \otimes I) \, (\Lambda \otimes {\cal I}) 
( (U_i \otimes I) \, \Phi \, (U_i^\dagger \otimes I) )  \, (U_i \otimes I) 
=  \!\!  \int \!\! d\mu(U) \, 
(U^\dagger \otimes I) \, (\Lambda \otimes {\cal I}) 
( (U \otimes I) \, \Phi \, (U^\dagger \otimes I) ) \, (U \otimes I) .
\label{eq:channel2bilateraltwirl}
\ee
We transform each side of the above equation in $3$ steps, turning 
the Choi matrix of the twirled channel into the bilateral twirl 
of an operator closely related to the Choi matrix of $\Lambda$.  
First, for the LHS of Eq.~(\ref{eq:channel2bilateraltwirl}), we apply 
the transpose trick 
$(U_i \otimes I)  \, \Phi \, (U_i^\dagger \otimes I) = 
 (I \otimes U_i^T)  \, \Phi \, (I \otimes U_i^*)$, 
where $T$ and $*$ denote the transpose and the complex conjugate
respectively.  Second, we commute the conjugation by $(I \otimes U_i^T)$
with $\Lambda \otimes {\cal I}$.  We apply similar manipulations
on the RHS of Eq.~(\ref{eq:channel2bilateraltwirl}). 
The equation becomes 
\be
\sum_{i=1}^k \, p_i \, 
(U_i^\dagger \otimes U_i^T) \, (\Lambda \otimes {\cal I}) 
(\Phi) \, (U_i \otimes U_i^*) 
=  \!\!  \int d\mu(U) \, 
(U^\dagger \otimes U^T) \, (\Lambda \otimes {\cal I})(\Phi) \, 
(U \otimes U^*)  \,.
\label{eq:channel2bilateraltwirl-2}
\ee
Third, we apply to Eq.~(\ref{eq:channel2bilateraltwirl-2})
the \emph{partial transpose} of the second system:  
for any $A_1, A_2 \in {\cal B}(\mathbb{C}^N)$, 
this linear map takes $A_1 \otimes A_2$ to $A_1 \otimes A_2^T$.  In 
particular, the partial transpose of 
$(I \otimes U_i^T) (\Phi)  (I \otimes U_i^*)
= \sum_{l,j=1}^N  |l\>\<j| \otimes (U_i^T |l\>\<j| U_i^*)$ 
is equal to 
$\sum_{l,j=1}^N  |l\>\<j| \otimes (U_i^\dagger |j\>\<l| U_i)
=(I \otimes U_i^\dagger) (\chi)  (I \otimes U_i)$ 
where $\chi = \sum_{l,j=1}^N  |l\>\<j| \otimes |j\>\<l|$ 
is the swap operator on $\mathbb{C}^N \otimes \mathbb{C}^N$.  
Eq.~(\ref{eq:channel2bilateraltwirl-2}) becomes 
\be
\sum_{i=1}^k \, p_i \, 
(U_i^\dagger \otimes U_i^\dagger) \, (\Lambda \otimes {\cal I}) 
(\chi) \, (U_i \otimes U_i) 
=  \!\!  \int d\mu(U) \, 
(U^\dagger \otimes U^\dagger) \, (\Lambda \otimes {\cal I})(\chi) \, 
(U \otimes U) 
\label{eq:channel2bilateraltwirl-3}
\ee
which is equivalent to 
\be {\cal T}_{\cE^\dagger}((\Lambda \otimes {\cal I})(\chi)) 
= {\cal T}_\mu( (\Lambda \otimes {\cal I})(\chi)). 
\label{eq:channel2bilateraltwirl-4}
\ee  
(In the above, we have used the fact $d\mu(U^\dagger) = d\mu(U)$.) 
Altogether, the transpose trick, the commutation, and the partial transpose 
transform Eq.~(\ref{eq:channel2bilateraltwirl}) concerning the equality of 
the Choi-matrices of the two channel twirls for $\Lambda$ 
into Eq.~(\ref{eq:channel2bilateraltwirl-4}) establishing the equality of 
the two bilateral twirls of the matrix $(\Lambda \otimes {\cal I})(\chi)$.  

It remains to apply Eq.~(\ref{eq:channel2bilateraltwirl-4}) to a set of
carefully chosen $\Lambda$'s to show that 
${\cal T}_{\cE^\dagger}(A) = {\cal T}_\mu(A)$ 
for a basis $\{A\}$ of the input space.
This will show that $\cE^\dagger$ implements the full
bilateral twirl.  By Corollary \ref{cor:dagger}, $\cE$ also implements
the full bilateral twirl and the proof will be completed.

We consider $\Lambda$'s with a specific form.  
Let ${\cal R}$ be the completely randomizing map on ${\cal B}(\mathbb{C}^N)$, 
i.e., ${\cal R}(\rho) = (\tr\rho) I/N$ for all $\rho \in {\cal B}(\mathbb{C}^N)$.  
Note that $J({\cal R}) = (I \otimes I)/N$.  
Consider any bounded linear map $\tilde{\Lambda}$ that is trace
preserving and for which $J(\tilde{\Lambda})$ is Hermitian 
(the latter property is called hermiticity preserving).
Then, for sufficiently small, positive, 
$\lambda$, $\Lambda = (1-\lambda) {\cal R} + \lambda \tilde{\Lambda}$ 
has positive semidefinite Choi-matrix (because the Choi-matrix of 
${\cal R}$ is proportional to the identity), and is therefore completely positive.  
Furthermore, $\Lambda$ is linear and trace-preserving.  
So, $\Lambda$ is a quantum channel.  
When we apply Eq.~(\ref{eq:channel2bilateraltwirl-4}) to such 
$\Lambda$'s, the ${\cal R}$ terms cancel out (because 
$({\cal R} \otimes {\cal I})(\chi) = (I \otimes I)/N$ which is 
invariant under either bilateral twirl).  Therefore, 
Eq.~(\ref{eq:channel2bilateraltwirl-4}) holds for all linear, 
trace and hermiticity preserving maps $\tilde{\Lambda}$ (which are
easier to construct than quantum channels).  

We are ready to show that ${\cal T}_{\cE^\dagger}(A) = {\cal
  T}_\mu(A)$ for a basis $\{A\}$ of the input space.
We take $A = H_l \otimes H_j$ where $\{H_l\}_{l=1}^{d^2}$ is a basis 
for ${\cal B}(\mathbb{C}^N)$ with the following additional properties:
\vspace*{-10pt}
\begin{itemize}
  \setlength\itemsep{-0.3em}
  \item[(1)] Each $H_l$ is Hermitian.
  \item[(2)] $H_1 = I/\sqrt{N}$.
  \item[(3)] $\tr(H_l H_j) = \delta_{lj}$. 
    In particular, $H_{l}$ is traceless for $l>1$. 
  \item[(4)] The swap operator has a simple representation 
in this basis,  
\be
     \chi = \sum_{l=1}^{d^2} H_l \otimes H_l \,. 
\ee
\end{itemize}
Such basis exists for all $N$.  When $N=2^n$, $H_l$ can be taken to be
proportional to the Pauli matrices (see Eq.~(\ref{eq:pauli-swap}) for the last 
condition).  For general $N$, we show in
Appendix \ref{appendix:gell-mann} that the \emph{generalized Gell-Mann
  matrices} can be used to construct such $H_l$'s.

We will verify that 
${\cal T}_{\cE^\dagger}(H_l \otimes H_j) = {\cal T}_\mu(H_l \otimes H_j)$ 
for all $1 \leq l,j \leq d^2$ by considering four cases.
First, the equality is immediate for $l=j=1$.  
Second, for each $1 < j \leq d$ consider $\tilde{\Lambda}_{1j}$ 
defined by $\tilde{\Lambda}_{1j} (H_1) = H_1 + H_j$,  
and $\tilde{\Lambda}_{1j} (H_l) = 0$ for all $l \neq 1$.
$\tilde{\Lambda}_{1j}$ is trace-preserving 
since each $H_{l}$ is traceless for $l>1$.  
Furthermore, $(\tilde{\Lambda}_{1j} \otimes {\cal I})(\chi) = (H_1 + H_j) \otimes H_1$
and partial transposing the second system gives $J(\tilde{\Lambda}_{1j})$, 
which implies $\tilde \Lambda_{1j}$ is Hermitian. 
Therefore, we can apply Eq.~(\ref{eq:channel2bilateraltwirl-4}) to 
$\tilde{\Lambda}_{1j}$ and conclude 
${\cal T}_{\cE^\dagger}(H_j \otimes H_1) = {\cal T}_\mu(H_j \otimes H_1)$. 
Third, because of the symmetry of the bilateral twirl, 
${\cal T}_{\cE^\dagger}(H_1 \otimes H_j) = {\cal T}_\mu(H_1 \otimes H_j)$. 
Fourth, let  $1 < j \leq l \leq d$ and consider $\tilde{\Lambda}_{jl}$ 
such that $\tilde{\Lambda}_{jl}(H_1) = H_1$, 
$\tilde{\Lambda}_{jl}(H_j) = H_l$, and $\tilde{\Lambda}_{jl}(H_{j'}) = 0$
for all $j' \neq 1$ and $j' \neq j$.  With arguments similar to the 
second case, ${\cal T}_{\cE^\dagger}(H_l \otimes H_j) 
= {\cal T}_\mu(H_l \otimes H_j)$.  This completes the proof.  

\vspace*{1.5ex}

\noindent Corollary \ref{cor:dagger}(b): We have established the 
equivalence between (3) and (4), thus, by Corollary \ref{cor:dagger}(a)
$\cE$ implements the full channel twirl if and only if $\cE^\dagger$ does.  

\end{proof}

\section{Short proof for (4) $\Rightarrow$ (3) in Lemma \ref{lem:equivdef}}
\label{appendix:shortproof}

Here, we consider the special case when $\cE = \{p_i, U_i\}$ is an
ensemble with Clifford unitaries and $N=2^n$.  We will show that if
$\cE$ implements the full channel twirl then it implements the full
bilateral twirl.

The proof relies on several definitions in Section~\ref{sec:pauli-mixing}.
We will show that if $\cE$
implements the full channel twirl then it is necessarily Pauli mixing,
and the rest follow from Lemma~\ref{lemma:pauli-mixing-to-2design}.
Consider an ensemble $\cE = \{p_i, U_i\}$ with Clifford unitaries
$U_i$ such that $\mathbb{E}_\cE(\Lambda) = \mathbb{E}_\mu(\Lambda)$
for all quantum channels $\Lambda$.  Take an arbitrary Pauli 
matrix $P \in \cQ_n$
with $P \neq I$ and an overall phase so that $P =
P^\dagger$.  Let $\Lambda(\rho) = P \rho P^\dagger$.  On one hand,
$\mathbb{E}_\cE(\Lambda)(\rho) = \sum_{i=1}^k \, p_i \, (U_i^\dagger P 
\, U_i) \, \rho \, (U_i^\dagger P \, U_i)^\dagger$.  
On the other
hand, $\mathbb{E}_\mu(\Lambda)(\rho) = (1-\lambda) \rho +
\tfrac{\lambda}{2^{2n}{-}1} \sum_{Q \in \cQ_n\backslash\{I\}} Q
\rho Q^{\dagger}$ for some $0 \leq \lambda \leq 1$.  
Note that for each $i$, $U_i$ is in the Clifford group 
so $U_i^\dagger P \, U_i$ is a Pauli matrix.  
Thus, we have two Kraus representations for the same twirled channel, 
both with Kraus operators in the quotient Pauli group $\cQ_n$, which is a basis
for $2^n \times 2^n$ matrices over $\mathbb{C}$.  Invoking 
Theorem 8.2 of \cite{NC00} concerning the degrees of freedom over 
these Kraus operators, the $i$-th term of $\mathbb{E}_\cE(\Lambda)$ 
can only contribute to $Q$ in $\mathbb{E}_\mu(\Lambda)$ 
if and only if $U_i^\dagger P \, U_i$ is equivalent to $Q$ in $\cQ_n$ 
(see Section~\ref{sec:pauli-mixing}) .  Finally, 
each $Q \neq I$ appears with equal weight in
$\mathbb{E}_\mu(\Lambda)(\rho)$, thus the distribution 
$\{p_i,U_i^\dagger P \, U_i\}$ is uniform over $\cQ_n \backslash \{I\}$.

\section{Construction of the basis $\{H_l\}$}
\label{appendix:gell-mann}

We want $\{H_l\}_{l=1}^{d^2}$ to be a basis 
for ${\cal B}(\mathbb{C}^N)$ with the following additional properties: 
\begin{itemize}
  \setlength\itemsep{-0.3em}
  \item[(1)] Each $H_l$ is Hermitian.
  \item[(2)] $H_1 = I/\sqrt{N}$.  
  \item[(3)] $\tr(H_l H_j) = \delta_{lj}$. 
    In particular, $H_{l}$ is traceless for $l>1$.  
  \item[(4)] The swap operator $\chi = \sum_{l=1}^{d^2} H_l \otimes H_l$. 
\end{itemize}

We use the \emph{generalized Gell-Mann matrices} for the construction. 
Let $H_1 = I/\sqrt{N}$.  For $l=2,\cdots,N$, let $H_l =
D_l/\sqrt{l(l{-}1)}$ where $D_l$ is a diagonal matrix with
$(D_l)_{1,1} = \cdots = (D_l)_{(l{-}1,l{-}1)} = 1$, $(D_l)_{l,l} =
-(l-1)$, and $(D_l)_{j,j} = 0$ for $l{+}1 \leq j \leq d$.  For $1 \leq
j_1 < j_2 \leq d$, let $X_{j_1,j_2} = (|j_1\>\<j_2| +
|j_2\>\<j_1|)/\sqrt{2}$, $Y_{j_1,j_2} = i(-|j_1\>\<j_2| +
|j_2\>\<j_1|)/\sqrt{2}$.  Let $\{ H_{d+1},\cdots,H_{d^2} \} = \{
X_{j_1,j_2}, Y_{j_1,j_2}\}_{1 \leq j_1 < j_2 \leq d}$ with any
ordering.  Then, $\{H_l\}_{l=1}^{d^2}$ span ${\cal B}(\mathbb{C}^N)$, each
$H_l$ is Hermitian, and $\tr(H_l H_j) = \delta_{lj}$.  
Finally, the expression for the swap operator $\chi$ can be verified by 
checking that each of the $d^4$
matrix entries on the RHS has the value given by the LHS.  
The verification involves routine arithmetic, each off-diagonal
element involves only 2 terms, and the diagonal elements can be
expressed as simple telescopic sums.

\section{Elementary proof that Pauli mixing implies a unitary 2-design}

\label{appendix:pmi2d}


\def\rchi{\raisebox{0.5ex}{$\chi$}}

\noindent {\bf Lemma \ref{lemma:pauli-mixing-to-2design}.}
  {\it Let $\cE$ be an ensemble of Clifford unitaries and
  $\cE_\cQ$ be as defined in Section~\ref{sec:pauli-mixing}. 
  If $\cE$ is Pauli mixing,
  then $\cE_\cQ$ implements the full bilateral twirl.}

\begin{proof} 
  The goal is to show that ${\cal T}_{\cE_\cQ}(\rho) =
  {\cal T}_\mu(\rho)$ for all density matrices $\rho$. 
  Note that both ${\cal T}_{\cE_\cQ}$ 
  and ${\cal T}_\mu$ are linear transformations on $2^{2n} \times 
  2^{2n}$ matrices.  Therefore, it suffices to show that 
  ${\cal T}_{\cE_\cQ}$ and ${\cal T}_\mu$ act 
  identically on a basis for these matrices.  We consider 
  a basis that contains the identity matrix $I_{2n}$ and the swap 
  operator $\rchi_{2n}$ acting on $2n$ qubits, completed with 
  matrices $M$ trace orthonormal to $I_{2n}$ and $\rchi_{2n}$ 
  (i.e., $\mathrm{Tr}(I_{2n}M) = \mathrm{Tr}(\rchi_{2n}M) = 0$).
  We will prove the following three claims:
  \begin{enumerate} 
    \item ${\cal T}_\mu (I_{2n}) = {\cal T}_{\cE_\cQ}(I_{2n}) = I_{2n}$, 
    \item ${\cal T}_\mu (\rchi_{2n}) = {\cal T}_{\cE_\cQ}(\rchi_{2n}) = \rchi_{2n}$, and 
    \item if 
      $\mathrm{Tr}(I_{2n}M) = \mathrm{Tr}(\rchi_{2n}M) = 0$,
      then ${\cal T}_\mu (M) = {\cal T}_{\cE_\cQ}(M) = \boldsymbol{0}$.  
  \end{enumerate}
Recall from Eqs.\ (\ref{eq:btwirle}) and (\ref{eq:btwirlf}) that 
  \begin{align*}
{\cal T}_\mu(\rho) & = \int d\mu(U) \; U\otimes U \; \rho \; 
U^{\dag}\otimes U^{\dag} ~~{\rm and}  \\
{\cal T}_{\cE_\cQ}(\rho) & = \sum_{i,j}p_i2^{-2n} (U_iR_j \otimes
U_iR_j) \; \rho \; (R_j^{\dag}U_i^{\dag}\otimes R_j^{\dag}U_i^{\dag}) \,.
  \end{align*}
It follows that the first claim holds trivially.  Furthermore, since 
$\rchi_{2n} (A \otimes B) \rchi_{2n} = B \otimes A$, or equivalently, $\rchi_{2n} (A
\otimes B) = (B \otimes A) \rchi_{2n}$, the second claim follows.

To prove the third claim, it suffices to show 
${\cal T}_{\cE_\cQ}(M) = \boldsymbol{0}$.  
This is because, for any $2^{2n} \times 2^{2n}$ matrices $\tilde{M}$, 
${\cal T}_\mu(\tilde{M}) = {\cal T}_\mu ({\cal T}_{\cE_\cQ}(\tilde{M}))$.  
In turns, this is due to the fact that 
$\forall \, V {\in} \, \ug_{2^n}, \forall \tilde{M},
{\cal T}_\mu(\tilde{M}) = {\cal T}_\mu(V {\otimes} V \tilde{M} 
V^\dagger {\otimes} V^\dagger)$;
applying the last identity to each unitary in $\cE_\cQ$ and
invoking linearity gives the desired result.  

We now show that 
${\cal T}_{\cE_\cQ}(\tilde{M}) = \boldsymbol{0}$. 
We make a 
crucial observation that 
$\rchi_2 = \frac{1}{2}(I\otimes I + X\otimes X +
Y\otimes Y + Z\otimes Z)$, and thus
\begin{equation}
\rchi_{2n} = \frac{1}{2^n} \sum_{R_l \in \cQ_n} R_l \otimes R_l \,.
\label{eq:pauli-swap}
\end{equation} 
Now, we use the fact that $\cQ_n$ is a basis for $2^n \times 2^n$
matrices to write $M = \sum_{ab}
\alpha_{ab} R_a \otimes R_b$ for some $\alpha_{ab} \in \mathbb{C}$.
We take $R_0 = I_n\in\mathcal{Q}_n$, so, the two conditions on $M$ can
be rephrased as $\alpha_{00} = 0$ and $\sum_{a} \alpha_{aa} = 0$.
By linearity, we focus on analyzing 
${\cal T}_{\cE_\cQ}(R_a \otimes R_b)$ for any $(a,b) \neq (0,0)$.
Note that 
\bea
   {\cal T}_{\cE_\cQ}(R_a \otimes R_b) 
 = \sum_{i} \, p_i \, (U_i \otimes U_i) 
      \left[ 
      \sum_j 2^{{-}2n} \, (R_j \otimes R_j) \; 
      (R_a \otimes R_b) \; (R_j^\dagger \otimes R_j^\dagger) 
      \right]
      (U_i^\dagger \otimes U_i^\dagger) \,.
\label{eq:mixing2btwirl}
\eea
If $a \neq b$, $\exists c$ such that $R_c$ commutes with $R_a$ and
anticommutes with $R_b$.  So, 
\bea
  2 \; \sum_j (R_j \otimes R_j) \; (R_a \otimes R_b) \; 
  (R_j^\dagger \otimes R_j^\dagger) \hspace*{56ex}
\nonumber
\\
  =  
  \sum_j (R_j \otimes R_j) \; (R_a \otimes R_b) \; 
  (R_j^\dagger \otimes R_j^\dagger) + 
  \sum_j (R_j R_c \otimes R_j R_c) \; (R_a \otimes R_b) \; 
  (R_c^\dagger R_j^\dagger \otimes R_c^\dagger R_j^\dagger) ~= ~\boldsymbol{0} 
\nonumber
\eea
and ${\cal T}_{\cE_\cQ}(R_a \otimes R_b) = \boldsymbol{0}$.
If $a = b$, 
\be\sum_j 2^{{-}2n} \, (R_j \otimes R_j) \; (R_a \otimes R_a) 
\; (R_j^\dagger \otimes R_j^\dagger) = (R_a \otimes R_a) \,.
\ee
Substituting the above into Eq.~(\ref{eq:mixing2btwirl}) and using 
the fact that $\cE$ is Pauli mixing, we obtain 
\be
{\cal T}_{\cE_\cQ}(R_a \otimes R_a) 
= \frac{1}{2^{2n}{-}1} \sum_{R_j \in \cQ_n \backslash \set{I}} R_j \otimes R_j 
= T
\, 
\label{eq:swap}
\ee
for a matrix $T$ independent of $a$.   
Putting all the pieces together, 
\be{\cal T}_{\cE_\cQ}(M) = \sum_{ab}
\alpha_{ab} {\cal T}_{\cE_\cQ}(R_a \otimes R_b) = \sum_{a}
\alpha_{aa} {\cal T}_{\cE_\cQ}(R_a \otimes R_a) = \left( \sum_{a}
\alpha_{aa} \right) T =  
\boldsymbol{0}.\ee  

\end{proof}

\section{Lower bounds for size and depth of unitary 2-designs}
\label{appendix:optimality}



Let $\cE = \set{p_i,U_i}_{i=1}^k$ be any exact unitary $2$-design on
$n$ qubits.  We show that a high probability set of the unitaries have
size $\Omega(n)$ and depth $\Omega(\log n)$, assuming a universal gate
set consisting of 1- and 2-qubit gates.  Both proofs invoke only
Definition~\ref{def:btwirl}, and they apply to unitary $2$-designs that
approximate the exact operation under Definition \ref{def:2queryu} or
\ref{def:btwirl} in the diamond norm.

Suppose the circuit for $U_i$ acts nontrivially on $s_i$ qubits. 
We will show that 
$\sum_{i=1}^k p_i s_i \geq n/2$, so, on average the circuit size is 
at least $n/2$.  
Since $\cE$ implements the full bilateral twirl, the quantum operation
$\rho \rightarrow \sum_{i=1}^k p_i U_i \rho U_i^\dagger =
\frac{I}{2^n}$ is the complete randomization map on $n$ qubits.  For
each $j$, consider the input $|0\>\<0|$ on the $j$th qubit.  Since the
output on the $j$th qubit is $\frac{I}{2}$, with probability at least
$\frac{1}{2}$, it has been acted on by one of the $U_i$'s.  
Define a matrix with rows labeled by $i=1,\cdots,k$, and columns 
labeled by $j=1,\cdots,n$, and the $(i,j)$ entry is $p_i$ if 
$U_i$ acts nontrivially on qubit $j$.  The above argument implies that each
column sums to at least $1/2$.  Also, by definition, the $i$th row sums 
to $s_i p_i$.  
The total of the row sums is equal to the total of the column sums, 
so, $\sum_{i=1}^k p_i s_i \geq n/2$, as claimed.  
Furthermore, consider the set $\mathscr{S} = \{i: s_i < n/4\}$.  
If $\sum_{i \in \mathscr{S}} p_i > 2/3$, $\sum_{i=1}^k p_i s_i \not\geq n/2$,
so, with probability at least $1/3$, the circuit has size at least
$n/4$.

\comment{
     \ qubit 
prob    1  2  ... n 

p1      1  0  0 

p2      1

..

pk      0  1 

multiply the p_i's to the column: 
each column sum >= 1/2
total col sum >= n/2
total row sum >= n/2 = w1 p1 + w2 p2 + ... + wk pk

n/4 * 2/3 + n * 1/3 = n/2.
}

For the lower bound on the depth, consider the bilateral twirl 
${\cal T}_{\cE}$ applied to the matrix $Z \otimes I^{\otimes n-1}
\otimes Z \otimes I^{\otimes n-1}$,
\be
{\cal T}_{\cE}(Z \otimes I^{\otimes n-1} \otimes Z \otimes I^{\otimes n-1}) 
= \sum_{i=1}^k p_i \; ( U_i (Z \otimes I^{\otimes n-1}) U_i^\dagger ) 
               \otimes ( U_i (Z \otimes I^{\otimes n-1}) U_i^\dagger ) \,.
\label{nonum1} 
\ee
%
%
Express each $U_i (Z \otimes I^{\otimes n-1}) U_i^\dagger$ as a linear
combination of Pauli matrices, and define the weight $t_i$ to be the
number of qubits that are acted on nontrivially by at least one of the
terms.  Since each gate interacts with at most two-qubits, if the depth of
the circuit for $U_i$ is $d_i$, then $d_i \geq \log t_i$.  We now show 
that most $U_i (Z \otimes I^{\otimes n-1}) U_i^\dagger$ have weight 
$t_i \geq n/2$.  

From Appendix \ref{appendix:pmi2d}, 
\be
{\cal T}_{\cE}(Z \otimes I^{\otimes n-1} \otimes Z \otimes I^{\otimes n-1}) 
= \frac{1}{2^{2n}{-}1} \sum_{R_j \in \cQ_n \backslash \set{I}} R_j \otimes R_j 
\label{nonum2} 
\ee
The fraction of $R_j$'s with weight less than $n/2$ is equal to 
$$4^{-n} \sum_{l=0}^{\lfloor n/2 \rfloor} {n \choose l} \, 3^l 
\leq 4^{-n} \sum_{l=0}^{\lfloor n/2 \rfloor} {n \choose l} \, 3^{n/2} 
\leq 4^{-n} \cdot \frac{1}{2} \cdot 2^n \cdot 3^{n/2}
\approx 0.866^n \,. $$   
Let $\mathscr{T} = \{i: t_i \geq n/2\}$.  
Then $\sum_{i \in \mathscr{T}} p_i \rightarrow 1$ but in particular 
$\sum_{i \in \mathscr{T}} p_i \geq 1/2$, because 
otherwise, the RHS of Eq.~(\ref{nonum1}) and (\ref{nonum2}) cannot be equal.

\comment{
(# Paulis with weight at least n/2) / 4^n -> 1
 because (# Paulis with weight at most n/2) < = 1/2 * 2^n * 3^{n/2}

Prob(Z1 -> Pauli with weight at least n/2) -> 1

Number of Paulis with weight at most n/2 
  =  \sum_{j=0}^{n/2} {n \choose j} 3^j 
\leq \sum_{j=0}^{n/2} {n \choose j} 3^{n/2} 
\approx 1/2 * 2^n * 3^{n/2} 

}

\section{Pauli group permutations are uniquely induced}
\label{appendix:uniqueness}

\begin{lemma}\label{lemma:appendix-uniqueness}
Suppose that unitaries $U$ and $V$ have the property that they induce the same permutation on the Pauli group so that, for all $a, b \in\{0,1\}^n$, 
\begin{align}\label{eq:appendix-conjugate-by-U}
U X^aZ^{b} U^{\dagger} \equiv V X^aZ^{b} V^{\dagger},
\end{align}
where $\equiv$ means equal up to a global phase that can be a function of $a$ and $b$.
Then $V = U X^cZ^d$ for some $c, d \in \{0,1\}^n$ (up to a global phase).
(Here $a$ and $b$ are binary strings, as opposed to elements of $\GF(2^n)$, so we do not require the notation $\lceil a \rceil$ and $\lfloor b \rfloor$ that occurs in other sections.)
\end{lemma}

\begin{proof}
Note that Eq.~\eqref{eq:appendix-conjugate-by-U} is equivalent to 
\begin{align}\label{eq:appendix-alternate}
X^a Z^b(U^{\dag}V)(X^a Z^b)^{\dag} = \lambda_{a,b} U^{\dag}V
\end{align}
for all 
$a,b \in \{0,1\}^n$
where $\lambda_{a, b}$ is the global phase in Eq.~\eqref{eq:appendix-conjugate-by-U}.
We can express $U^{\dag}V$ as 
\begin{align}\label{eq:appendix-expression}
U^{\dag}V = \sum_{c,d \in \{0,1\}^n} \alpha_{c,d} X^c Z^d.
\end{align}
Recall that $X^aZ^b$ and $X^cZ^d$ either commute or anticommute, depending on the value of the symplectic inner product%
\footnote{The symplectic inner product is defined as $(a,b)\cdot(c,d) = (\oplus_{k=1}^{n} a_k d_k) \oplus (\oplus_{k=1}^{n} b_k c_k)$.}
 of $(a,b)$ and $(c,d)$ (they commute when $(a,b)\cdot(c,d) = 0$ and anticommute otherwise).
Using this fact and substituting Eq.~\eqref{eq:appendix-expression} into Eq.~\eqref{eq:appendix-alternate}, we obtain 
\begin{align}\label{eq:appendix-symplectic}
\sum_{c,d \in \{0,1\}^n} (-1)^{(a,b)\cdot(c,d)}\alpha_{c,d} X^c Z^d
=
\sum_{c,d \in \{0,1\}^n} \lambda_{a,b}\alpha_{c,d} X^c Z^d.
\end{align}
Since the Paulis $X^c Z^d$ are linearly independent, the coefficients must match.

We now show that at most one $\alpha_{c,d}$ can be nonzero.
Suppose two are nonzero:  $\alpha_{c_1,d_1} \neq 0 \neq \alpha_{c_2,d_2}$ for some $(c_1,d_1) \neq (c_2,d_2)$.
Then there exists $(a,b)$ such that $(a,b) \cdot (c_1,d_1) \neq (a,b) \cdot (c_2,d_2)$.
Then, from Eq.~\eqref{eq:appendix-symplectic}, we can deduce that
\begin{align}
(-1)^{(a,b)\cdot(c_1,d_1)} = \lambda_{a,b} = (-1)^{(a,b)\cdot(c_2,d_2)},
\end{align}
which is a contradiction.
Therefore there is a unique nonzero $\alpha_{c,d}$, which implies $V = \alpha_{c,d} U X^c Z^d$.
\end{proof}

\end{document}